\documentclass[11pt]{article}
\usepackage{amsmath,amsfonts,amssymb,amsthm,accents,mathtools,algorithm,algorithmic}
\usepackage{epsfig,url,cite,graphicx,graphics,subfigure,psfrag}
\usepackage{vmargin} 
\usepackage{fancyhdr}
\newcommand{\single}{\renewcommand{\baselinestretch}{1}\large\normalsize}
\newcommand{\double}{\renewcommand{\baselinestretch}{1}\large\normalsize}

\newcommand{\real}{\mathbb{R}}
\newcommand{\norm}[1]{\lVert#1\rVert} 
\newcommand{\expectv}{\mathbb{E}} 
\newcommand{\indicator}[1]{\mathbf{1}_{#1}} 
\newcommand{\probmeasureHzero}{P} 
\newcommand{\probmeasureHone}{Q} 
\newcommand{\disHzero}{p}
\newcommand{\disHone}{q} 
\newcommand{\genempir}[1]{\nu_{\noobs}(#1)} 

\newcommand{\ddimen}{d} 
\newcommand{\asize}{L} 
\newcommand{\noobs}{n} 
\newcommand{\partindex}{i} 
\newcommand{\coeff}{c} 
\newcommand{\lbound}{m} 
\newcommand{\ubound}{M} 
\newcommand{\covcomp}{\alpha} 
\newcommand{\compconst}{A} 
\newcommand{\marp}{\kappa} 
\newcommand{\marconst}{K} 
\newcommand{\denlbound}{c} 
\newcommand{\denubound}{C} 
\newcommand{\RDPdepth}{J} 
\newcommand{\compp}{\beta} 
\newcommand{\aindex}{i} 
\newcommand{\obsindex}{i} 

\newcommand{\mapc}{\gamma} 
\newcommand{\mapcopt}{\psi^{*}} 
\newcommand{\mapout}{\phi} 
\newcommand{\mapsopt}{\mapout^{*}} 
\newcommand{\mapempest}{\widehat{\mapout}_{\noobs}} 
\newcommand{\PCclass}{\textnormal{PC}(\beta,m,M,L)} 
\newcommand{\candclass}{\Phi} 

\newcommand{\region}{R} 
\newcommand{\optcell}{A^{*}} 
\newcommand{\cell}{S} 
\newcommand{\partition}{\pi} 
\newcommand{\partclass}{\Pi_{\region}} 

\newcommand{\abbrefdiv}[1]{D(#1)} 
\newcommand{\empirfdiv}[1]{D_{\noobs}(#1)} 
\newcommand{\kl}[2]{D_{\text{\slshape{KL}}}(#1\Vert#2)}

\newcommand{\inputs}{x} 
\newcommand{\inputsv}{x} 
\newcommand{\inputrv}{X} 

\newcommand{\auxrv}{Z}

\newtheorem{theorem}{Theorem}
\newtheorem*{defin}{Definition}
\newtheorem{lemma}{Lemma}

\numberwithin{example}{section}

\setpapersize{A4}
\setmargins{1in}{.75in}{6.25in}{9in}{14pt}{10pt}{20pt}{30pt}
\title{\vspace*{-20pt}\textbf{Quantization via Empirical Divergence Maximization}}
\author{Michael~A.~Lexa\\ \url{amlexa@gmail.com}}
\date{5 Nov 2011, Revised 30 Apr 2012}

\begin{document}
\sloppy\double
\maketitle

\begin{abstract}
Empirical divergence maximization (EDM) refers to a recently proposed strategy for estimating $f$-divergences and likelihood ratio functions.
This paper extends the idea to empirical vector quantization where one seeks to empirically derive quantization rules that maximize the Kullback-Leibler divergence between two statistical hypotheses.
We analyze the estimator's error convergence rate leveraging Tsybakov's margin condition and show that rates as fast as $n^{-1}$ are possible, where $n$ equals the number of training samples.
We also show that the Flynn and Gray algorithm can be used to efficiently compute EDM estimates and show that they can be efficiently and accurately represented by recursive dyadic partitions.
The EDM formulation have several advantages.
First, the formulation gives access to the tools and results of empirical process theory that quantify the estimator's error convergence rate.
Second, the formulation provides a previously unknown derivation for the Flynn and Gray algorithm.
Third, the flexibility it affords allows one to avoid a small-cell assumption common in other approaches.
Finally, we illustrate the potential use of the method through an example.
\end{abstract}

\section{Introduction}\label{sect:intro}
In statistical learning theory, empirical risk minimization is a standard technique whereby classifiers are formed from empirical data~\cite{devroye-gyorfi-lugosi1996}.
The idea is simple enough: when the underlying probability distributions characterizing the data are unknown, classifiers are found by minimizing an empirical form of the risk (probability of error) over some specified class of classifiers.
The technique is well-understood and has been generalized to include various cost criteria and problem settings.
In its generalized form, empirical risk minimization is sometimes referred to as $M$-estimation (the $M$ standing for minimization or maximization)~\cite{vandegeer2007}.

Recently, Nguyen, Wainwright and Jordan~\cite{nguyen_wainwright_jordan2010} applied $M$-estimation to the estimation of $f$-divergences (the Kullback-Leibler (KL) divergence~\cite{kullback1959} in particular) and to bounded likelihood ratio functions.
In this paper, we build on their ideas and develop a method for computing empirical quantization rules by maximizing the KL divergence.
We call the method \emph{empirical divergence maximization} (EDM) in deference to its similarity to empirical risk minimization and because the name is simple and descriptive.
The proposed formulation leads to an entirely different algorithm for computing the estimators than that employed in~\cite{nguyen_wainwright_jordan2010}, and the convergence rates reported here incorporate a margin condition not included in~\cite{nguyen_wainwright_jordan2010} that shows when fast convergence is possible.

As the name suggests, the criterion used in EDM is the KL divergence, a well-known information theoretic quantity that has enjoyed a prominent and long-standing place in both theory and practice.
Applications are numerous and range from detection and estimation problems~\cite{poor_thomas1977,chamberland_veeravalli2003,lhler_etal2004} to texture retrieval in image databases~\cite{do_vetterli2002} and from the study of neural coding~\cite{johnson_gruner2001} to linguistic problems~\cite{cai_kulkarni_verdu}.
Roughly speaking, the KL divergence quantifies the dissimilarity of two probability density functions (pdfs) and is therefore often regarded as a ``distance'', although it is not a distance metric.
Stein's Lemma~\cite[p. 309] {cover_thomas1991} fundamentally links the divergence to hypothesis testing by relating it to the decay rate of different error probabilities. 
In fact, the divergence equals the optimal asymptotic error decay rate of a Neyman-Pearson test.
Thus increasing the divergence between two statistical hypotheses generally increases their discriminability.

The use of the KL divergence in quantization problems dates back nearly four decades~\cite{poor_thomas1977}.
In that time, various problem settings have been investigated including scalar, vector, and distributed quantization~\cite{poor_thomas1977,kassam1988,longo_lookabaugh_gray1990,gupta_hero2003}.
Until recently, however, most results addressing this type of quantization assumed full knowledge of the probability distributions of interest and did not explicitly address empirical designs.
Moreover, those works in the quantization literature most closely related to the present paper~\cite{gupta_hero2003,lazebnik_raginsky2009} invoke a small-cell assumption that forces partitions, designed to maximize the divergence, to resemble nearest neighbor partitions even when such partitions cannot well-approximate theoretically optimal partitions (see Fig.~\ref{fig:lr}). 
Because of its flexibility, however, the EDM approach overcomes this shortcoming.

In~\cite{lazebnik_raginsky2009}, Lazebnik and Raginsky study a conceptually similar quantization problem to the one considered here, but the differences between the approaches are substantial.
For example, their information loss criterion is a difference of mutual informations, and while related to the KL divergence, this criterion measures a different quantity than the divergence loss studied here.
Their work is also placed in a machine learning setting where the data and the quantization values (labels) are jointly distributed and both play integral roles in their information criterion.
In this paper, the quantization values play a secondary role in the computation of our estimator.

To formalize the problem, let $\probmeasureHzero\text{~and~}\probmeasureHone$ be two probability measures defined on the probability space $([0,1]^{\ddimen},\mathcal{B})$, where $\mathcal{B}$ denotes the usual Borel $\sigma$-algebra and $d\geq1$.
Let $\disHzero$ and $\disHone$ denote the density functions of  $\probmeasureHzero\text{~and~}\probmeasureHone$ with respect to Lebesgue measure and assume $P$ and $Q$ are absolutely continuous with respect to one another.
Then any quantization rule $\mapc:\real^{\ddimen} \mapsto \{0,\dotsc,\asize-1\}$ that operates on a random vector $X$ (distributed according to $\probmeasureHzero$ or $\probmeasureHone$) induces the probability mass functions (pmfs),
 $\disHzero(\mapc)=(\disHzero_{0}(\mapc),\dotsc,\disHzero_{\asize-1}(\mapc))$ and $\disHone(\mapc)=(\disHone_{0}(\mapout),\dotsc,\disHone_{\asize-1}(\mapc))$, where $\disHzero_{\aindex}(\mapc)=\probmeasureHzero(\mapc(X)=\aindex)$ and similarly for $\disHone_{\aindex}(\mapc)$.
In this context, the KL divergence is defined as
\begin{equation*}
\kl{\disHzero(\mapc)}{\disHone(\mapc)}:=\sum_{\aindex=0}^{\asize-1} -\disHzero_{\aindex}(\mapc)\log{\bigg(\frac{\disHone_{\aindex}(\mapc)}{\disHzero_{\aindex}(\mapc)}\bigg)}.
\end{equation*}
In EDM, we maximize an empirical form of the KL divergence over some given class of quantization rules.
We therefore analyze an estimator of the form
\begin{equation*}
\widehat{\gamma}_{n}=\underset{\mapc\in\Gamma}{\arg\max} ~\empirfdiv{\mapc},
\end{equation*}
where $\empirfdiv{\mapc}$ represents an empirical KL divergence that is defined in Section~\ref{sec:EDM} (the subscript $n$ signifies that it is an empirical quantity that is based on $n$ samples from both $P$ and $Q$) and where $\Gamma$ denotes some class of quantization rules.
By design EDM estimators constructed rules $\widehat{\gamma}_{n}$ that induce maximally divergent pmfs, thereby best preserving the discriminability of $P$ and $Q$.
In other words, EDM estimators maximize the performance (in terms of KL divergence) of any downstream detector or classifier that operates on the quantized data.
The EDM formulation has several advantages:  (i) it readily permits the application of empirical process theory which in turn provides the tools to quantify the estimator's error decay rates; (ii) it naturally leads to the Flynn and Gray algorithm which efficiently computes the quantization rules; (iii) it provides a systematic derivation for the Flynn and Gray algorithm; and (iv) the flexibility in candidate function classes allows efficient representation of the quantization rules and overcomes the small-cell constraint.

\section{Empirical Divergence Maximization}\label{sec:EDM}
The form of $\empirfdiv{\mapc}$ is taken from recent work by Nguyen et al.~\cite{nguyen_wainwright_jordan2010} and relies on rewriting the convex function $-\log(\cdot)$ appearing in the definition of the KL divergence.
Throughout the paper, we hold the number of quantization levels $L$ fixed.

\subsection{Expressing divergence using convex conjugates}
The notion of a \emph{convex conjugate} is based on the observation that a curve can either be described by its graph or by an envelope of tangents~\cite{rockafellar1970}.
More concretely, a (closed) convex function $f:\real\mapsto\real$ can be described as the pointwise supremum of a collection of affine functions $h(t)=tt^*-\mu^*$ such that the set of all pairs $(t^*,\mu^*)$ lie within the epigraph of its convex conjugate $f^*(t^*)$, i.e., 
\begin{equation}\label{equ:convexdual}
f(t)= \sup_{t^*}~\{t^*t-f^*(t^*)\}
\end{equation}
where by duality the convex conjugate $f^*(t^*)$ of $f(t)$ is defined by
\begin{equation*}
f^*(t^*)= \sup_{t}~\{tt^*-f(t)\}.
\end{equation*}

Now suppose $\mapc$ is an arbitrary quantization rule defined on $[0,1]^d$,
\begin{equation}
\mapc(x) = \sum_{i=0}^{L-1} i \, \indicator{R_i}(x),~x\in[0,1]^d,
\end{equation}
where $\{\region_{\partindex}\}_{\partindex=0}^{\asize-1}$ is a collection of disjoint sets partitioning $[0,1]^d$ and $\indicator{R_i}(\cdot)$ denotes the indicator function.
Using~\eqref{equ:convexdual}, we can write the divergence between the pmfs induced by $\mapc$ as
\begin{subequations}
\begin{align}
\kl{\disHzero(\mapc)}{\disHone(\mapc)}&=\sum_{\aindex=0}^{\asize-1} \disHzero_{\aindex}(\mapc)f\bigg(\frac{\disHone_{\aindex}(\mapc)}{\disHzero_{\aindex}(\mapc)}\bigg) \label{equ:klconvexconj1} \\
&=\sum_{\aindex=0}^{\asize-1} \disHzero_{\aindex}(\mapc)\cdot
  \sup_{t^*}\Bigl\{t^* \frac{\disHone_{\aindex}(\mapc)}{\disHzero_{\aindex}(\mapc)} - f^*(t^*)\Bigl\}, \label{equ:klconvexconj2}
\end{align}
\end{subequations}
where $f(t)=-\log(t)$ for $t>0$ and $+\infty$ otherwise.
Calculating the convex conjugate, one finds
\begin{equation*}
f^*(t^*)=\begin{cases}-1-\log(-t^*)~~\text{if~~} t^*<0 \\
+\infty ~~\text{if~~} t^*\geq 0.
\end{cases}
\end{equation*}
Substituting this expression into~\eqref{equ:klconvexconj2}, we have the following expressions for the KL divergence
\begin{align*}
&\kl{\disHzero(\mapc)}{\disHone(\mapc)} \\
&\qquad =\sum_{\aindex=0}^{\asize-1} \disHzero_{\aindex}(\mapc)
  \sup_{t^*_{\partindex}\in\real^{-}}\Bigl\{t^*_{\partindex} \frac{\disHone_{\aindex}(\mapc)}{\disHzero_{\aindex}(\mapc)} +1 +\log(-t^*_{\partindex})\Bigl\} \nonumber\\
&\qquad = \sum_{\aindex=0}^{\asize-1} \disHzero_{\aindex}(\mapc)
  \negmedspace \sup_{\coeff_{\region_{\partindex}}\in\real^{+}}\Bigl\{\log(\coeff_{\region_{\partindex}}) - \coeff_{\region_{\partindex}} \frac{\disHone_{\aindex}(\mapc)}{\disHzero_{\aindex}(\mapc)} + 1\Bigl\}\nonumber\\
&\qquad = \negmedspace 1\negthinspace+\negmedspace\sum_{\aindex=0}^{\asize-1}
  \sup_{\coeff_{\region_{\partindex}}\in\real^{+}}\negmedspace\Bigl\{\negthinspace\probmeasureHzero(\region_{\partindex})\log(\coeff_{\region_{\partindex}}) - \coeff_{\region_{\partindex}}\, \probmeasureHone(\region_{\partindex})\negmedspace \Bigl\},
\end{align*}
where in the second step we let $\coeff_{\region_{\partindex}}=-t^*_{\partindex}$, and in the last step use the fact that
  $\disHzero_{\partindex}(\mapc)=\probmeasureHzero(\region_{\partindex}), \region_{\partindex}=\{x:\mapc(x)=\partindex\}$.
The validity of last expression is easily verified by differentiating it with respect to $\coeff_{\region_{\partindex}}$ and solving for the maximizers.
By defining the piecewise constant function
\begin{equation}\label{equ:piecerule}
\mapout(x):=\sum_{\partindex=0}^{\asize-1}\coeff_{\region_{\partindex}} \indicator{\region_{\partindex}}(x),\quad \coeff_{\region_{\partindex}}\in\real^{+}, x\in[0,1]^d
\end{equation}
we can write $\kl{\disHzero(\mapc)}{\disHone(\mapc)}$ in integral form:
\begin{equation}\label{equ:convexconjkl}
1 + \sup_{\mapout} \left\{\int_{[0,1]^{\ddimen}} \log(\mapout)~d\probmeasureHzero - \int_{[0,1]^{\ddimen}} \mapout~d\probmeasureHone \right\},
\end{equation}
where the supremum is taken over all functions of the form~\eqref{equ:piecerule}.
Note that the $\phi$ which achieves the supremum depends on $P$, $Q$, and $\{R_i\}_{i=0}^{L-1}$.
Below, we restrict $\phi$ to lie within a (more) specific class of rules and define the proposed quantization rule estimator in terms of the empirical counterpart to~\eqref{equ:convexconjkl}.

In addition, note that unlike $\gamma$, $\phi$ does not map $[0,1]^{d}$ to a set of indices.
We nevertheless refer to both as quantization rules since $\phi$ only assumes $L$ real values.
Note also that in terms of KL divergence, $\phi$ determines $\gamma$, i.e., if $\phi$ is known, a quantization rule $\gamma: [0,1]^{d}\mapsto \{0,\dotsc,L-1\}$ can be defined that induces the same pmfs as $\phi$.
This fact becomes important for the algorithm described in Section~\ref{sect:algorithm}.

\subsection{Empirical estimator}\label{subsect:empest}
To define a function class for $\phi$, we first consider different ``labelings'' of a uniform partition of $[0,1]^d$.
For a given positive integer $J$, let $\pi_J$ denote a tesselation of $[0,1]^{d}$ by uniform hypercubes $S_k, k=0,\dotsc,2^{dJ}-1$.
To each cell $S_k$, we can associate one of $L$ labels $\{0,\dots,L-1\}$, and thus for each different labeling of $\pi_{J}$, we can define another partition, $\pi_{R}$, with cells $\{R_i\}_{i=0}^{L-1}$ described by
\begin{equation}
R_i=\bigcup_{k:\,\text{label}(S_k)=i} \negthickspace S_k, \quad i=0,\dotsc,L-1.
\end{equation}
Now, for a given partition $\pi_R$ and positive constants $\lbound>0$ and $\ubound<\infty$, denote by $\candclass_{\pi_{R}}(\asize,J,\lbound,\ubound)$ the set of all $L$-level piecewise constant functions defined on $\pi_{R}$ that are bounded and positive:
\begin{equation*}
\candclass_{\pi_{R}}(\asize,J,\lbound,\ubound)=\negthickspace \left\{\mapout(x)=\negthickspace \sum_{\partindex=0}^{\asize-1}\coeff_{\region_{\partindex}} \indicator{\region_{\partindex}}(x)\negthickspace : \lbound\leq\coeff_{\region_{\partindex}}\leq\ubound \negthinspace \right\}.
\end{equation*}
Letting $\Pi_{R}$ denote the set of all partitions $\pi_{R}$, or equivalently the set of all different labelings of $\pi_J$, we define the candidate class $\candclass(\asize, J,\lbound,\ubound)$ of our empirical quantizers as
\begin{equation}\label{equ:defcandclass2}
\candclass(\asize, J,\lbound,\ubound):= \bigcup_{\pi_R \in \Pi_R} \candclass_{\pi_{R}}(\asize,J,\lbound,\ubound).
\end{equation}

Letting $\{\inputrv_{\obsindex}^{\disHzero}\}_{\obsindex=1}^{\noobs}$ and $\{\inputrv_{\obsindex}^{\disHone}\}_{\obsindex=1}^{\noobs}$ be training data distributed according to $\disHzero$ and $\disHone$, respectively, we define the function $\empirfdiv{\mapout}$ as an empirical counterpart to~\eqref{equ:convexconjkl}
\begin{equation}\label{equ:empirKL}
\empirfdiv{\mapout}:=1+\frac{1}{\noobs}\sum_{\obsindex=1}^{\noobs} \log{\mapout(\inputrv_{\obsindex}^{\disHzero})}
 - \frac{1}{\noobs}\sum_{\obsindex=1}^{\noobs} \mapout(\inputrv_{\obsindex}^{\disHone})
\end{equation}
and define the proposed empirical quantization rule estimator as
\begin{equation}\label{equ:empirest2}
\mapempest:=\underset{\mapout\in\candclass(\asize,J,\lbound,\ubound)}{\arg\max} ~\empirfdiv{\mapout}.
\end{equation}
$\mapempest$ is our \emph{empirical divergence maximization} (EDM) estimator.
Note $\empirfdiv{\mapout}$ is not in general a KL divergence; it can in fact be negative for some ${\mapout\in\candclass}$.
It is a consistent estimator, however, converging to the ``best in class'' estimator as $\noobs\rightarrow \infty$~\cite{nguyen_wainwright_jordan2010}.

\subsection{Best in class and optimal quantization rules}
The \emph{best in class estimate} $\mapsopt$ is that element in $\candclass$ that maximizes $\abbrefdiv{\mapout}$,
\begin{equation}\begin{split}\label{equ:knownpq_KLest}
&\mapsopt:=\underset{\mapout\in\candclass(\asize,J,\lbound,\ubound)}{\arg\max} ~\abbrefdiv{\mapout},\quad \text{where} \\
\abbrefdiv{\mapout}&:=1 + \int_{[0,1]^{\ddimen}} \log(\mapout)~d\probmeasureHzero - \int_{[0,1]^{\ddimen}} \mapout~d\probmeasureHone.
\end{split}
\end{equation}
Note that $\abbrefdiv{\mapout}$, as opposed to $\empirfdiv{\mapout}$, is not an empirical quantity; its definition requires full knowledge of the distributions $\probmeasureHzero$ and $\probmeasureHone$.

We take the theoretically optimal quantization rule $\mapcopt$ to be the rule that maximizes the divergence over a class of piecewise constant functions that has an assumed boundary regularity (the regularity conditions play a role in the convergence analysis in Section~\ref{subsect:approxerror}).
The class definition uses the notion of a locally constant function: a function $f:[0,1]^{\ddimen}\mapsto\real$ is \emph{locally constant} at a point $\inputs\in[0,1]^{\ddimen}$ if there exists $\epsilon>0$ such that for all $y\in[0,1]^{\ddimen}$, the condition $\norm{\inputs-y}<\epsilon$ implies $f(y)=f(\inputs)$.
\begin{defin}[PC class~\cite{ruicastro_phd_thesis2007}]
A function $f:[0,1]^{\ddimen}\mapsto\{\coeff_{\partindex}\}_{\partindex=0}^{\asize-1}, \coeff_{\partindex}\in\real^{+}$ is a positive-valued piecewise constant function with $\asize$ levels if it is locally constant at any point $x\in[0,1]^{\ddimen}\setminus B(f)$, where $B(f)\subset [0,1]^{\ddimen}$ is a boundary set satisfying $N(r)\leq \beta r^{-(\ddimen-1)}$ for all $r>0$.
Here, $\beta>0$ is a constant and $N(r)$ is the minimal number of balls of diameter $r$ that covers $B(f)$.
Furthermore, let $f$ be uniformly bounded on $[0,1]^{\ddimen}$, that is $\lbound\leq f(x)\leq \ubound$ for all $x\in[0,1]^{\ddimen}$, where $\lbound>0$ and $\ubound<\infty$.
The set of all piecewise constant functions $f$ satisfying the above conditions is denoted by $\PCclass$.
\end{defin}
\noindent In short, we consider $\PCclass$ to be a class of likelihood-ratio quantization rules that have well behaved boundaries.
The theoretically optimal quantization rule is thus defined to be
\begin{equation}\label{equ:mapcopt}
\mapcopt:=\underset{\psi\in\PCclass}{\arg\max} ~\abbrefdiv{\psi}.
\end{equation}
It is well-known that $\mapcopt$ can always be constructed by thresholding the likelihood ratio~\cite{tsitsiklis1993a}.
In other words, the optimal quantization rule $\mapcopt$ can always be chosen to be a piecewise constant function whose boundary sets are level sets of the likelihood ratio $\disHone(x)/\disHzero(x)$.

\section{Solving for the estimator}\label{sect:algorithm}
To find $\mapempest$ in~\eqref{equ:empirest2}, we employ a modified form of the Flynn and Gray algorithm~\cite{flynn_gray1987} that iteratively maximizes the divergence between two pmfs over a set of quantization rules.
The method directly follows from the EDM formulation (although it was not originally proposed in this context) and searches for an optimal cell labeling for a given partition where the number of cells is much larger than the number of quantization levels.

\subsection{The Flynn and Gray algorithm}
For independent and identically distributed random variables $\inputrv_{1},\dotsc,\inputrv_{\noobs}$, the \emph{empirical measure} of a set ${A\in[0,1]^{\ddimen}}$, denoted $\probmeasureHzero_{\noobs}(A)$, is the sample average
\begin{equation}\label{equ:sample_average}
\probmeasureHzero_{\noobs}(A)=\frac{1}{\noobs}\sum_{k=1}^{\noobs} \indicator{A}(\inputrv_{k}).
\end{equation}
The sample average of a function $g:[0,1]^{\ddimen}\mapsto \real$ can thus be written with respect to $\probmeasureHzero_{n}$ as an \emph{empirical expectation},
\begin{equation}
\probmeasureHzero_{\noobs}(g)=\frac{1}{\noobs}\sum_{k=1}^{\noobs} g(\inputrv_{k})=\int g~d\probmeasureHzero_{\noobs}.
\end{equation}
Using this notation, we rewrite~\eqref{equ:empirKL} as
\begin{equation}\label{equ:empirKL1}
\empirfdiv{\mapout}=1 + \int_{[0,1]^{\ddimen}} \log(\mapout)~d\probmeasureHzero_{\noobs} - \int_{[0,1]^{\ddimen}} \mapout~d\probmeasureHone_{\noobs},
\end{equation}
where $\mapout\in\candclass$.
For any fixed partition $\partition_{\region}\in\partclass$, $\empirfdiv{\mapout}$ is maximized by assigning $\phi(x)$ the values $\probmeasureHzero_{\noobs}(\region_{i})/\probmeasureHone_{\noobs}(\region_{i})$ for $x\in R_i$.
For this assignment choice, $\empirfdiv{\mapout}$ can be expressed as
\begin{equation}\label{equ:flynn_gray_setup}
\empirfdiv{\mapout}=1 + \sum_{i=0}^{L-1} \int_{\region_{i}} \log\biggl(\frac{\probmeasureHzero_{\noobs}(\region_{i})}{\probmeasureHone_{\noobs}(\region_{i})}\biggr)~d\probmeasureHzero_{\noobs} - \int_{\region_{i}} \frac{\probmeasureHzero_{\noobs}(\region_{i})}{\probmeasureHone_{\noobs}(\region_{i})}~d\probmeasureHone_{\noobs},
\end{equation}
and the estimator $\mapempest$ can now be found by searching over $\Pi_R$ for the partition that maximizes~\eqref{equ:flynn_gray_setup}.
The Flynn and Gray algorithm~\cite{flynn_gray1987} is a straightforward method which accomplishes this task.
To apply it, we rewrite~\eqref{equ:flynn_gray_setup} as
\begin{align}\begin{split}
\empirfdiv{\mapout}&=\sum_{i=0}^{L-1} \probmeasureHzero_{\noobs}(\region_{i}) \left[\log\biggl(\frac{\probmeasureHzero_{\noobs}(\region_{i})}{\probmeasureHone_{\noobs}(\region_{i})}\biggr)+1\right] \\ &\qquad + \probmeasureHone_{\noobs}(\region_{i}) \biggl(-\frac{\probmeasureHzero_{\noobs}(\region_{i})}{\probmeasureHone_{\noobs}(\region_{i})}\biggr)\end{split} \\
&=\sum_{i=0}^{L-1} \probmeasureHzero_{\noobs}(\region_{i}) a_{i} + \probmeasureHone_{\noobs}(\region_{i}) b_{i} \\
&=\sum_{i=0}^{L-1} \sum_{k\in I_i} \probmeasureHzero_{\noobs}(S_k) a_{i} + \probmeasureHone_{\noobs}(S_k) b_{i},\label{equ:flynn_gray_coreeq}
\end{align}
where $a_{i}=\log\bigl(\probmeasureHzero_{\noobs}(\region_{i})/\probmeasureHone_{\noobs}(\region_{i})\bigr)+1$, $b_{i}=-\probmeasureHzero_{\noobs}(\region_{i})/\probmeasureHone_{\noobs}(\region_{i})$, and $I_i$ is the index set ${\{k:\text{label}(S_k)=i\}}$.
The algorithm maximizes $\empirfdiv{\mapout}$ by iterating two steps: it first holds the set of weights $\{a_i\}$ and $\{b_i\}$ fixed and finds the labels for each cell $S_k\in \pi_J, k=0,\dotsc,2^{dJ}-1$ that maximizes~\eqref{equ:flynn_gray_coreeq}, and then holds the cell labels of $\pi_J$ fixed and updates the weights $\{a_i\}$ and $\{b_i\}$ using the probabilities $\probmeasureHzero_{\noobs}(\region_{i}), \probmeasureHone_{\noobs}(\region_{i})$, $i=0,\dotsc,L-1$ found from the first step.
Flynn and Gray showed these steps monotonically increase~\eqref{equ:flynn_gray_coreeq}, and since $\empirfdiv{\mapout}$ is upper bounded by $1-m+\log{M}$ (follows from the boundedness of $\phi$), the algorithm converges to a local maximum.
The algorithm returns a locally optimal labeling of $\pi_{J}$ and locally optimal weights from which $\mapempest$ can be determined:
\begin{equation}
\mapempest(x)=-b_i \quad \text{for~}x\in R_i.
\end{equation}
The algorithm is outlined in the panel entitled Algorithm~\ref{algor:flynn_gray}.
An advantage of the Flynn and Gray algorithm is that it avoids the exhaustive combinatoric search over all possible labelings by only needing to examine each cell $S_k$ once per iteration.
From experiments, it has been observed that for moderate sized partitions $\pi_J$ ($<2^{16}$ cells) and for $L<10$, the algorithm converges very quickly ($<30$ iterations). 

EDM provides a new derivation for the Flynn and Gray algorithm; however, it is interesting to note that it can also be based on the fact that
\begin{equation}\label{graycomment}
\kl{p}{q}= \sup_{\gamma}~ \kl{\disHzero(\mapc)}{\disHone(\mapc)}
\end{equation}
where the supremum is over all measurable quantization rules with an arbitrary number of levels~\cite{gray1990,poor1988}.
Because $\kl{p}{q}\geq \kl{\disHzero(\mapc)}{\disHone(\mapc)}$ for any quantization rule, one could use~\eqref{graycomment} to justify an approach similar to EDM and maximize $\kl{\disHzero(\mapc)}{\disHone(\mapc)}$ over a set of quantization rules for a fixed quantization level.
While this approach leads to similar (if not identical) estimators, EDM has the advantage of making a clear connection with empirical process theory which provides the theoretical tools to analyze the error convergence rate.

Note also that the original Flynn and Gray algorithm does not explicitly constrain the values of $\phi$ to lie within the range $[m,M]$.
However, to avoid computing unbounded estimates at any given iteration, we employ the K-T technique~\cite{krichevsky-trofimov1981} when computing  $P_n(S_k)$ and $Q_n(S_k)$, ${k=0,\dotsc,2^{dJ}-1}$.
This technique simply preloads each cell $S_k$ by one half before calculating the sample averages, thereby avoiding the possibility of computing zero probability estimates $P_n(R_i)$, $Q_n(R_i)$.
Thus instead of~\eqref{equ:sample_average}, one computes
\begin{equation}
P_n(S_k) = \frac{1}{2^{dJ-1}+n}\left(\frac{1}{2}+\sum_{k=1}^{\noobs} \indicator{S_k}(X_{k}^{p})\right)
\end{equation}
and likewise for $Q_n(S_k)$. 
Here, the choice of $1/2$ is not arbitrary; it is based on theoretical considerations of what \emph{a priori} distribution of the probabilities $P(S_k)$, $Q(S_k)$ influences the sample averages $P_n(S_k)$, $Q_n(S_k)$ the least~\cite{krichevsky-trofimov1981}.

\begin{algorithm}
\caption{Modified Flynn and Gray algorithm}
\label{algor:flynn_gray}
\double
\begin{algorithmic}[1]
\REQUIRE \parbox[t]{7cm}{L, J, n, empirical cell probabilities $P_n(S_k)$ and $Q_n(S_k)$, stopping threshold $\epsilon$ \vspace*{.2cm}}
\STATE Initialize iteration index $l=0$
\STATE Randomly label cells $S_k\in\pi_{J}$
\STATE Compute $P_n^{(l)}(R_i), Q_n^{(l)}(R_i),~i=0,\dotsc,L-1$
\STATE Initialize weights \\ $a^{(l)}_i=\log\bigl(P_n^{(l)}(R_i)/Q_n^{(l)}(R_i)\bigr)+1$ \\ $b^{(l)}_i=-P_n^{(l)}(R_i)/Q_n^{(l)}(R_i)$
\STATE Compute $D^{(l)}_n(\phi)$
\STATE Initialize intermediary divergence $\widetilde{D}_n(\phi)=0$
\WHILE{$(D^{(l)}_n(\phi)-\widetilde{D}_n(\phi))/D^{(l)}_n(\phi)>\epsilon$}
\STATE Find new label for each cell $S_k$ by computing \\ 
$i^{l+1}=\arg\max_{i\in\{0,\dotsc,L-1\}} P_n(S_k) a^{(l)}_{i} + Q_n(S_k) b^{(l)}_{i}$ \\ (hold weights fixed)
\STATE Update probabilities for $i=0,\dotsc,L-1$ \\
$P_n^{(l+1)}(R_i)=\sum_{k:\text{label}(S_k)=i} P_n(S_k)$ \\
$Q_n^{(l+1)}(R_i)=\sum_{k:\text{label}(S_k)=i} Q_n(S_k)$ \\
(includes K-T preloading if necessary)
\STATE Compute intermediary divergence \\
$\widetilde{D}_n(\phi)=\sum_{i=0}^{L-1} P_n^{(l+1)}(R_i) a^{(l)}_{i} + Q_n^{(l+1)}(R_i) b^{(l)}_{i}$
\STATE Update weights \\
$a^{(l+1)}_i=\log\bigl(P_n^{(l+1)}(R_i)/Q_n^{(l+1)}(R_i)\bigr)+1$ \\ $b^{(l+1)}_i=-P_n^{(l+1)}(R_i)/Q_n^{(l+1)}(R_i)$
\STATE Compute new divergence \\
$D^{(l+1)}_n(\phi)=\sum_{i=0}^{L-1} P_n^{(l+1)}(R_i) a^{(l+1)}_{i} + Q_n^{(l+1)}(R_i) b^{(l+1)}_{i}$
\STATE $l=l+1$
\ENDWHILE
\ENSURE \parbox[t]{7cm}{$\mapempest$ (locally optimal labels of $\pi_J$ and weights $\{a_i,b_i\}$), $D_n(\mapempest)$}
\end{algorithmic}
\end{algorithm}

In short, one solves for an EDM estimator based upon the training data $\{\inputrv_{\obsindex}^{\disHzero}\}, \{\inputrv_{\obsindex}^{\disHone}\}$ by first computing the sample averages $P_n(S_k), Q_n(S_k)$ for each cell $S_k\in \pi_{J}$ and then providing these probabilities as input to the Flynn and Gray algorithm.
The algorithm is applied to two numerical examples in Section~\ref{sect:app}.

\subsection{Recursive dyadic partitions}\label{subsect:RDP}
Because $\Phi$ is based on a uniform dyadic partition, any EDM estimate $\mapempest$ can be viewed as a piecewise constant function supported on a \emph{recursive dyadic partition} (RDP).
RDPs are a systematic class of partitions that have proven to be effective in function estimation and classification problems~\cite{ruicastro_phd_thesis2007,scott_nowak2006}.
Their usefulness stems from their ability to adapt to boundaries (including $\PCclass$), thus allowing efficient computation of estimators and concise encoding of estimates.
In the present context, RDPs are important because they allow efficient encoding of $\mapempest$, and their properties are key in the approximation error analysis presented in Section~\ref{subsect:approxerror}.

RDPs are partitions composed of quasi-disjoint sets%
\footnote{Two sets are quasi-disjoint if and only if their intersection has Lebesgue measure zero.} 
whose union equals the entire space $[0,1]^{\ddimen}$.
A RDP is any partition that can be constructed using only the following rules~\cite{ruicastro_phd_thesis2007}:
\begin{enumerate}
\item $\{[0,1]^{\ddimen}\}$ is a RDP.
\item Let $\partition=\{\cell_{0},\ldots,\cell_{k-1}\}$ be a RDP, where $\cell_{\partindex}=[u_{\partindex 1},v_{\partindex 1}]\times\ldots\times [u_{\partindex \ddimen},v_{\partindex \ddimen}]$.
Then 
\begin{equation*}
\partition'= \{\cell_{0},\dotsc, \cell_{\partindex-1}, \cell_{\partindex}^{0},\dotsc, \cell_{\partindex}^{(2^{\ddimen}-1)}, \cell_{\partindex+1},\ldots,\cell_{k-1}\}
\end{equation*} 
is a RDP, where
 $\{\cell_{\partindex}^{0},\dotsc,\cell_{\partindex}^{(2^{\ddimen}-1)}\}$ is obtained by dividing the hypercube $\cell_{\partindex}$ into $2^{\ddimen}$ quasi-disjoint hypercubes of equal size.
Formally, let $q\in\{0,\dotsc,2^{\ddimen-1}\}$ and $q=q_1q_2\dotsc q_{\ddimen}$ by the binary representation of $q$.
Then
\end{enumerate}
\begin{align*}
\cell_{\partindex}^{(q)}=&\Biggl[u_{\partindex 1}+\frac{v_{\partindex 1}-u_{\partindex 1}}{2}q_{1},
  v_{\partindex 1}+\frac{u_{\partindex 1}-v_{\partindex 1}}{2}(1-q_{1})\Biggr] \times \\
&\ldots \times \Biggl[u_{\partindex \ddimen}+\frac{v_{\partindex \ddimen}-u_{\partindex \ddimen}}{2}q_{\ddimen},
  v_{\partindex \ddimen}+\frac{u_{\partindex \ddimen}-v_{\partindex \ddimen}}{2}(1-q_{\ddimen})\Biggr].
\end{align*}
We say a RDP has maximal depth $J$ if the side length of its smallest hypercube equals $2^{-J}$.
Fig.~\ref{fig:RDP} illustrates a RDP approximating an elliptical boundary.
\begin{figure}
\centerline{\includegraphics[width=5cm]{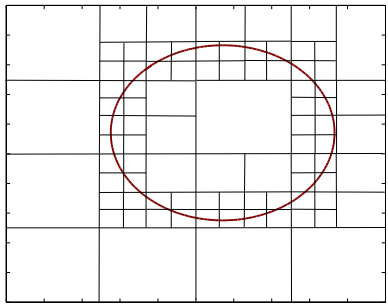}}
\caption{\small An example two-dimensional RDP ($\RDPdepth=3$).}
\label{fig:RDP}
\end{figure}

It should be clear RDPs describe tree structures where the root node is the entire space $[0,1]^d$ and the leaf nodes represent the different cells comprising the RDP.
Each branch can have different depths and thus the cells can have different sizes.
This property allows a RDP to have larger cells in locations where the function value is constant and smaller cells where the values change (around boundaries). 
The combination of the systematic tree structure and the partition's adaptivity allow the estimator to be efficiently encoded, that is, the number of bits necessary to map an observation to its quantized value can be done efficiently~\cite{scott_nowak2006}.

For a fixed estimator $\mapempest$ (or a fixed labeling of $\pi_J$), a RDP can be easily constructed by repeating step 2 above (starting with the whole space), but only producing a split if the cells $S_k \in \pi_J$ on either side of the split, but within the hypercube of interest, have different values (labels).

\section{Error Decay Rates}\label{sec:error_rates}
We gauge the quality of $\mapempest$ by characterizing the decay rate of the estimation and approximation errors.
\emph{Estimation error} is defined as the difference $\abbrefdiv{\mapsopt}-\abbrefdiv{\mapempest}$ and quantifies the error caused by computing $\mapempest$ without knowledge of $\disHzero$ and $\disHone$.
As the number of samples $n$ increases, the estimation error decreases at a rate (exponent of $n$) that depends on the complexity of $\Phi$ and on the properties of $p$ and $q$.
\emph{Approximation error} is defined as $\abbrefdiv{\mapcopt}-\abbrefdiv{\mapsopt}$ and arises in cases where $\mapcopt\notin\candclass$.
To quantify its decay, we think of the candidates rules $\phi\in\Phi$ as being supported on RDPs and the rate of decay in terms of the depth parameter $J$.

We begin in a standard fashion with two basic inequalities that follow from the definitions of $\mapsopt$ and $\mapempest$: ${\abbrefdiv{\mapsopt}-\abbrefdiv{\mapempest}\geq 0}$ and ${\empirfdiv{\mapsopt}-\empirfdiv{\mapempest}\leq 0}$.
They imply that the estimation error is upper bounded by a difference of empirical processes
\begin{align*}
0&\leq \abbrefdiv{\mapsopt}-\abbrefdiv{\mapempest}\\
&\leq-[(\empirfdiv{\mapsopt}-\abbrefdiv{\mapsopt})-(\empirfdiv{\mapempest}-\abbrefdiv{\mapempest})] \\
&=-(\genempir{\mapsopt}-\genempir{\mapempest})/\sqrt{\noobs},
\end{align*}
where the second inequality results from adding and subtracting $\empirfdiv{\mapsopt}$ and $\empirfdiv{\mapempest}$, and where $\genempir{\mapc}=\sqrt{\noobs}(\empirfdiv{\mapc}-\abbrefdiv{\mapc})$.
Adding the approximation error to both sides of the inequality bounds the total error by the two component errors.
\begin{align}\label{equ:fundaineqal}
\begin{split}
0&\leq \underset{\text{total error}}{\underbrace{\abbrefdiv{\mapcopt}-\abbrefdiv{\mapempest}}} \\
 &\leq -\underset{\text{upper bound on est. error}}{\underbrace{(\genempir{\mapsopt}-\genempir{\mapempest})/\sqrt{\noobs}}} + \underset{\text{approx. error}}{\underbrace{\abbrefdiv{\mapcopt}-\abbrefdiv{\mapsopt}}}
\end{split}
\end{align}
We use this equation and examine the estimation and the approximation errors separately, giving the final rate result for the expected total error.

\subsection{Estimation Error}\label{subsect:esterror}
Let $\expectv_p$ and $\expectv_q$ denote expectation operator with respect to $P$ and $Q$.
Then, by writing $\rvert\genempir{\mapempest}-\genempir{\mapsopt}\lvert$ as
\begin{equation*}
\begin{split}
&\Bigg\vert \frac{1}{\noobs}\sum_{\obsindex=1}^{\noobs}  (\log{\mapout(\inputrv_{\obsindex}^{\disHzero})}-\log{\mapsopt(\inputrv_{\obsindex}^{\disHzero})})- \expectv_{\disHzero} \big[\log{\mapout(\inputrv)}\\ 
&-\log{\mapsopt(\inputrv)}\big] + (\mapout(\inputrv_{\obsindex}^{\disHone})-\mapsopt(\inputrv_{\obsindex}^{\disHone})) - \expectv_{\disHone} \big[\mapout(\inputrv)-\mapsopt(\inputrv)\big] \Bigg\vert,
\end{split}
\end{equation*}
it is clear that for a given quantization rule $\mapout$, the empirical averages above converge almost surely to their respective values by the strong law of large numbers.
But because $\mapempest$ can potentially be any element in $\candclass$, any characterization of the convergence rate must hold uniformly over $\candclass$.
It is well-known that uniform rates of convergence depend on the complexity of the function class from which the empirical estimators are drawn~\cite{vandegeer2000}.
Here we use the notion of \emph{bracketing entropy} to characterize complexity of $\candclass$.
Roughly speaking, the bracketing entropy of a function class $\mathcal{G}$ equals the logarithm of the minimum number of function pairs that upper and lower bound (bracket) all the members in $\mathcal{G}$ to within some tolerance $\delta$ and with respect to some norm (a precise definition can be found in~\cite[p. 16]{vandegeer2000}).
We denote the bracketing entropy by $H_B(\delta,\mathcal{G},L_2(\probmeasureHzero))$ and say $\mathcal{G}$ has \emph{bracketing complexity} $\covcomp>0$ if $H_B(\delta,\mathcal{G},L_2(\probmeasureHzero))\leq \compconst \delta^{-\covcomp}$ for all $\delta>0$ and for some constant $A>0$.
Because the members of $\Phi$ are uniformly bounded, it can be shown that $\Phi$ has bracketing complexity $\alpha=1$~\cite{michaellexa_phdthesis2008}.
This fact is incorporated into Theorem~\ref{thm:esterrorconverge} below; however, the proof of the theorem given in Appendix~\ref{app:esterror} assumes the bracketing complexity lies between zero and two. 
The proof therefore yields a slightly more general result than that stated. 

We now introduce two conditions on $p$ and $q$. 
The first simply states that $p$ and $q$ are uniformly bounded.

\noindent\textbf{Condition 1.}
Assume $\denlbound\leq\disHzero(x),\disHone(x)\leq \denubound \text{~for all~} x\in[0,1]^{\ddimen}$, $\denlbound>0$, $\denubound<\infty$.

The second is a condition introduced by Mammen and Tsybakov~\cite{mammen_tsybakov1999,tsybakov2004} and involves a key parameter $\kappa$ that provides insight into when fast convergence rates are possible (i.e., rates faster than $\noobs^{-1/2}$).
The condition arises in a slightly different form in function estimation and Bayesian classification problems, and within these contexts, it can be related to the behavior of $p$ and $q$ near a boundary of interest%
\footnote{In Bayesian classification and function estimation, the condition is known as the \emph{margin condition}}.
For example, in Bayesian classification, small $\kappa$ implies a ``steep'' regression function at the Bayes decision boundary and thus easier classification; large $\kappa$ implies a ``flat'' transition and harder classification.
Van de Geer~\cite{vandegeer2007} describes $\kappa$ as an ``identifiability'' parameter in the sense that it characterizes how well $\phi\in\candclass$ can be distinguished from $\mapcopt$.
Because $\mapcopt$ is determined by $P$ and $Q$, this condition is ultimately a condition on these underlying distributions.

\noindent\textbf{Condition 2.}
There exists constants $\marconst>0$ and $\marp\geq 1$ such that for all $\mapout \in\candclass$,
\begin{equation}\label{equ:margincond}
\abbrefdiv{\mapcopt}-\abbrefdiv{\mapout}\geq \norm{\mapcopt-\mapout}_{L_2}^{\marp}/\marconst.
\end{equation}

If $\kappa$ is small (close to 1) for a given $P$ and $Q$, the difference $\abbrefdiv{\mapcopt}-\abbrefdiv{\mapout}$ is larger for $\phi\in\candclass$ close to $\mapcopt$ (those $\phi$ such that $\norm{\mapcopt-\mapout}_{L_2}\leq1$) compared to those distributions having larger $\kappa$ values.
Intuitively, this means that such $\phi$ are more distinguishable from $\mapcopt$ for those distributions satisfying Condition 2 with small $\kappa$ compared to those distributions satisfying Condition 2 with larger $\kappa$ values, where distinguishability is measured in terms of divergence loss $\abbrefdiv{\mapcopt}-\abbrefdiv{\mapout}$%
\footnote{Note that the distinguishability in terms of divergence loss is intimately connected with how $\mapempest$ is computed: since $\empirfdiv{\mapout}$ is a surrogate of $\abbrefdiv{\mapout}$, maximizing $\empirfdiv{\mapout}$ over $\mapout\in\candclass$ is a surrogate for maximizing $\abbrefdiv{\mapout}$ over $\mapout\in\candclass$, or equivalently minimizing $\abbrefdiv{\mapcopt}-\abbrefdiv{\mapout}$.}. 
The following result shows that $\kappa$ effectively characterizes this aspect of the problem, and like for Bayesian classification and estimation, is a key parameter for the error convergence rate.

\begin{theorem}[Estimation error]\label{thm:esterrorconverge}
Let $\mapempest$, $\mapsopt$, and $\mapcopt$ be as defined in~\eqref{equ:empirest2},~\eqref{equ:knownpq_KLest} and~\eqref{equ:mapcopt} respectively.  Suppose Conditions 1 and 2 are met for some constants $\denlbound, \denubound, \marconst,\text{~and~} \marp$.
Then for any $0<\epsilon<1$ we have
\begin{align*}\label{equ:vandegeerlem}
\abbrefdiv{\mapcopt}&-\expectv\abbrefdiv{\mapempest} \leq \Bigl(\frac{1+\epsilon}{1-\epsilon}\Bigr) \\ 
& \biggl[\textnormal{const}(\denlbound,\denubound,\marconst,\marp)~ \noobs^{-\frac{\marp}{2\marp-1}}+ \abbrefdiv{\mapcopt}-\abbrefdiv{\mapsopt}\biggr],
\end{align*}
for sufficiently large $\noobs$ where $\textnormal{const}(\denlbound,\denubound,\marconst,\marp)$ is a decreasing function of $\epsilon$.
\end{theorem}
The proof is provided in Appendix~\ref{app:esterror} (see also~\cite{mlexa2010a}) and directly follows from results of van de Geer\cite[pp. 206-207]{vandegeer2007}~\cite{vandegeer2000} and Mammen and Tsybakov~\cite{mammen_tsybakov1999}.

Depending on $P$ and $Q$, the decay rate of the estimation error can be as fast as $\noobs^{-1}$ ($\kappa=1$) and no worse than $\noobs^{-1/2}$ ($\kappa=\infty$).
In particular, if for a given $P$ and $Q$, the approximation error is nonzero (which is commonly the case in quantization problems), we have
\begin{align*}
\abbrefdiv{\mapcopt}-\abbrefdiv{\mapout}&\geq \abbrefdiv{\mapcopt}-\abbrefdiv{\mapsopt} \\
&\geq \text{const} \cdot \frac{\norm{\mapcopt-\mapout}_{L_2}}{M}
\end{align*}
where the first inequality follows from the definition of $\mapsopt$ and the second from the fact that $0\leq \tfrac{\norm{\mapcopt-\mapout}_{L_2}}{M} \leq 1$ for all $\phi\in\candclass$.
Thus with a nonzero approximation error, Condition 2 can be met with $\kappa=1$ and the rate $n^{-1}$ is achievable. 
This situation is common in quantization problems because it is unusual in practice for the level sets of a likelihood ratio function to coincide with a RDP for a fixed depth $J$.
For example, consider the simple scenario where $p(x)$ is a zero-mean unit variance Gaussian density function and $q(x)$ is a Laplace density function (also zero-mean and unit variance), $x\in\real$. 
With $L=8$ and $J=6$, the approximation error is $0.002$ and hence we expect a rate of $\mathcal{O}(n^{-1})$.
Fig.~\ref{fig:1D_gauss-laplace_rates} confirms this result experimentally by plotting $\abbrefdiv{\mapcopt}-\expectv\abbrefdiv{\mapempest}$ as a function of $n$.
(For each value of $n$, Gaussian and Laplacian data were generated and $\mapempest$ was computed using the Flynn and Gray algorithm.)
The black dashed curve on the left hand plot is shown for reference and equals $40n^{-1}+0.002$. 
\begin{figure*}
\begin{minipage}{.49\linewidth}
\centerline{\includegraphics[width=7cm]{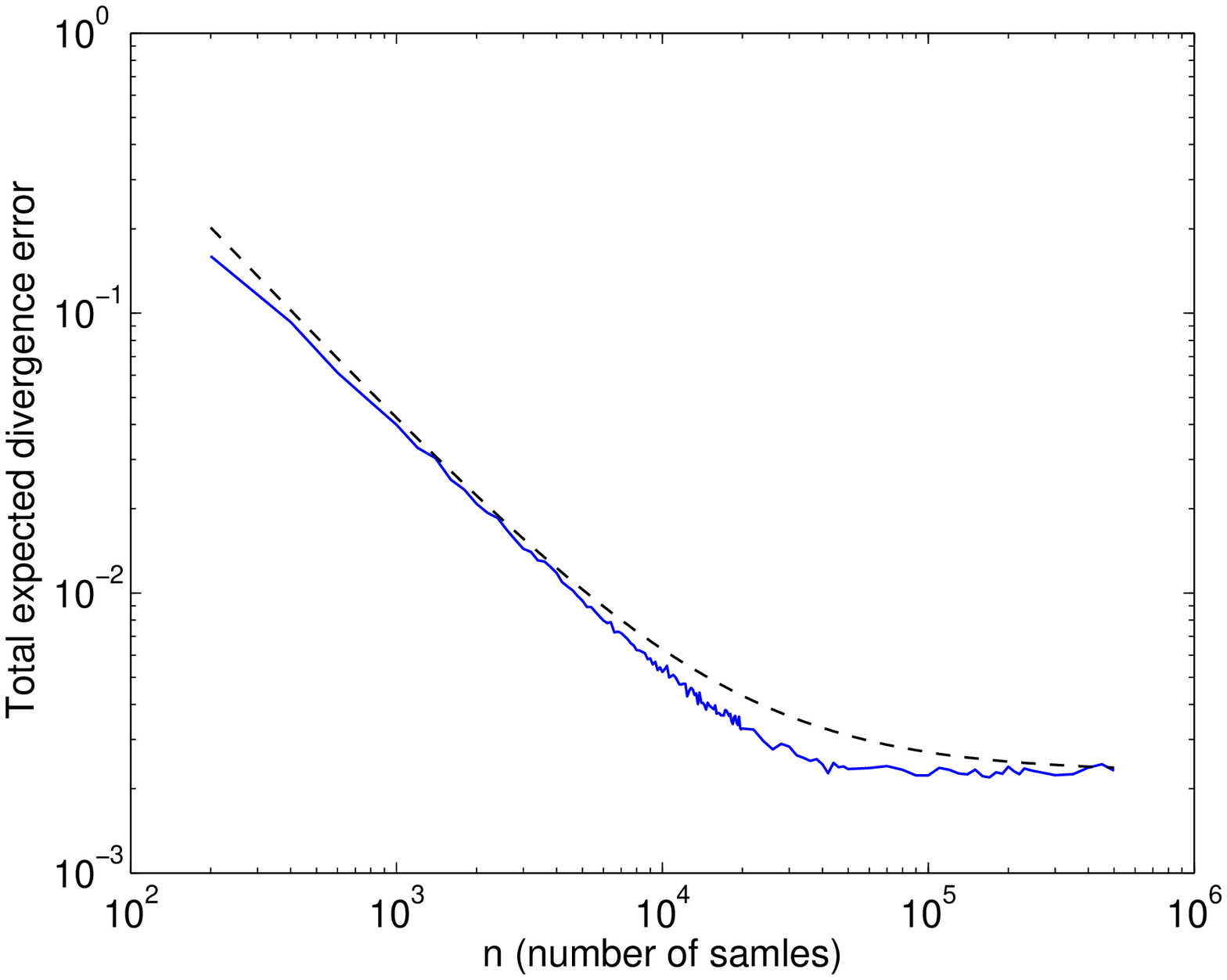}}
\end{minipage}
\begin{minipage}{.49\linewidth}
\centerline{\includegraphics[width=7cm]{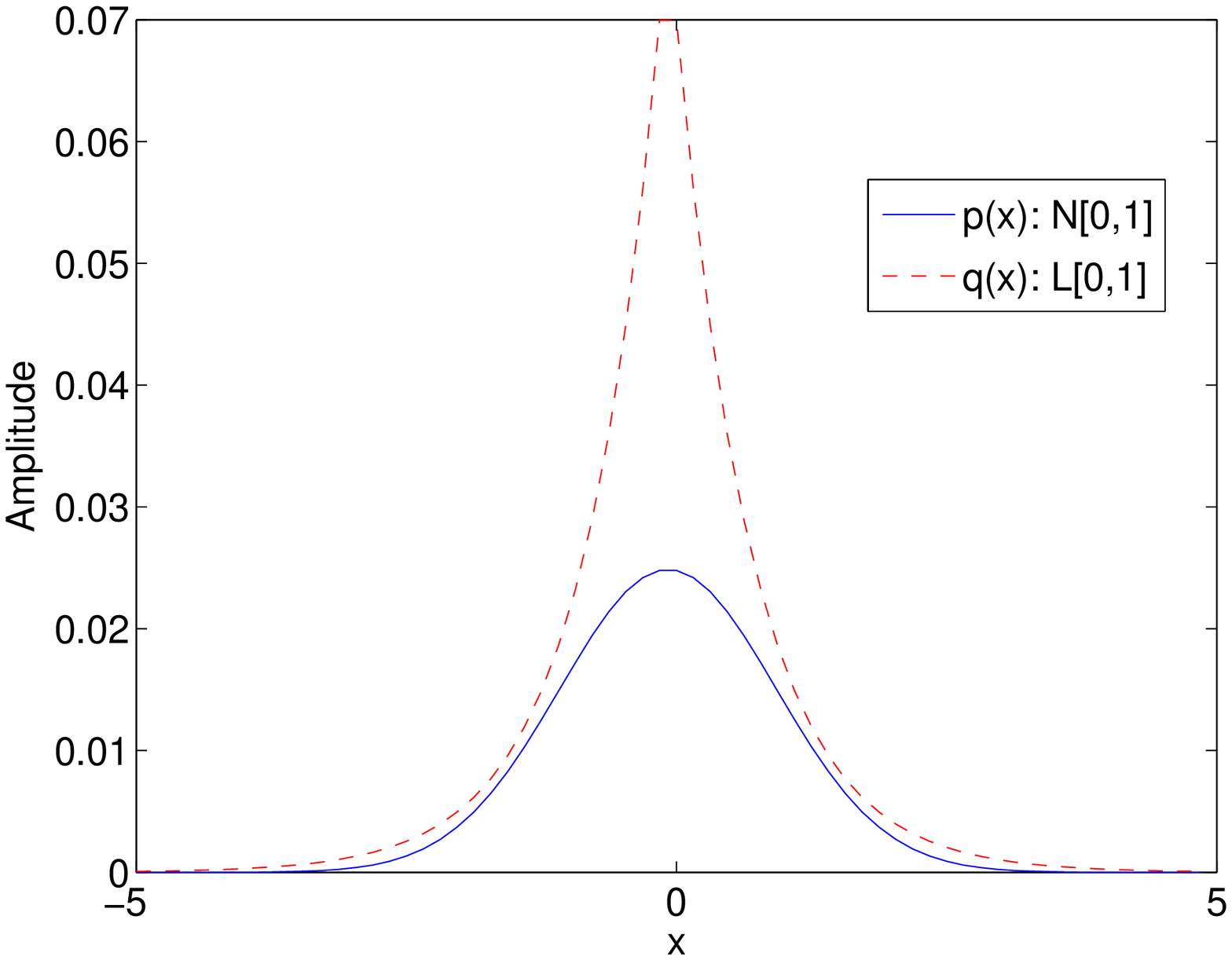}}
\end{minipage}
\caption{Left: Total expected divergence loss $\abbrefdiv{\mapcopt}-\expectv\abbrefdiv{\mapempest}$ plotted as a function of the number of training samples $n$ (solid curve) for the case $L=8$, $J=6$ where $p(x)$ is a zero-mean unit variance Gaussian density and $q(x)$ is a zero-mean unit variance Laplace density. The dashed (black) curve is $\mathcal{O}(n^{-1})$, thus for this example, we have a fast rate of decay.  Right: Probability density functions.}
\label{fig:1D_gauss-laplace_rates}
\end{figure*}

In contrast, if $P$ and $Q$ are such that $\mapcopt\in\candclass$, then Condition 2 is only met with $\kappa=2$, and therefore the resulting decay rate is $\mathcal{O}(n^{-2/3})$.
To derive this result, consider the subset of quantization rules $\mapout\in\candclass$ that share the same partition associated with $\mapcopt$.
Recalling~\eqref{equ:piecerule}, we can in this case write $\mapcopt$ as
\begin{equation*}
\mapcopt(\inputsv)=\sum_{i=0}^{\asize-1} \mapcopt_{i}\indicator{R_{i}}(\inputsv),
\end{equation*}
where $\{\mapcopt_{i}\}$ are the levels of $\mapcopt$.
Then by~\eqref{equ:knownpq_KLest} and the fact that $\mapcopt_i=P(R_i)/Q(R_i)$, we have
\begin{align*}
\abbrefdiv{\mapcopt}&=1+\int \log{(\mapcopt)}~d\probmeasureHzero-\int\mapcopt~d\probmeasureHone \\
&=\int\log{\mapcopt}~d\probmeasureHzero \\
&=\int\mapcopt \log{(\mapcopt)}~d\probmeasureHone \\
&= \sum_{i=0}^{\asize-1} \int_{\region_i} \mapcopt_{i} \log{(\mapcopt_{i})}~d\probmeasureHone.
\end{align*}
Now, because $(\cdot)\log{(\cdot)}$ is differentiable and continuous on the range of the positive real numbers, we can by Taylor's Theorem~\cite{fulks1978} expand $\mapcopt_{i}\log{\mapcopt_i}$ on $R_{i}$ around $c_{R_i}$ for each $i$ to obtain
\begin{equation}\label{equ:taylor1}
\abbrefdiv{\mapcopt}=\sum_{i} \int_{R_i} \mapcopt_{i} \log{(c_{R_i})} +\mapcopt_{i} - c_{R_i} + \mathcal{R}_{i}~d\probmeasureHone,
\end{equation}
where $\mathcal{R}_{i}=\dfrac{1}{2\delta_{i}}(\mapcopt_{i}-c_{R_i})^{2}, i=0,\dotsc,\asize-1$, are the Taylor remainders of the expansions with $\delta_{i}$ lying in between $\mapcopt_{i}$ and $c_{R_i}$.
By adding and subtracting $ \int \log{(\mapout)}~dP$, ~\eqref{equ:taylor1} can be rearranged to yield
\begin{align}
\abbrefdiv{\mapcopt}-\abbrefdiv{\mapout}&=\sum_{i}\int_{R_i}\mathcal{R}_{i}~dQ \nonumber \\
&\geq \dfrac{c}{2\ubound} \sum_{i}\int_{R_i} (\mapcopt_{i}-c_{R_i})^{2}~d\inputsv \\
&= \dfrac{c}{2\ubound} \norm{\mapcopt-\mapout}_{L_2}^{2}
\end{align}
where the inequality follows from replacing $\delta_i$ with $\ubound$ in the remainder term and using the fact that $q$ is lower bounded by $c$ (Condition 1).
Thus, when there is no approximation for the given distributions $P$ and $Q$ (and for given values of $J$, $L$, $m$, and $M$) the best guaranteed convergence rate is $\mathcal{O}(n^{-2/3})$.
Intuitively, this is reasonable since among those $\mapout$ that share the same partition as $\mapcopt$, it is harder to distinguish $\mapcopt$ compared to the case where there is a nonzero approximation error.

Nguyen, Wainwright and Jordan reported a similar result to Theorem~\ref{thm:esterrorconverge} in~\cite{nguyen_wainwright_jordan2010} .
In their investigation, they used an empirical estimator of the same form as~\eqref{equ:empirest2}, but did not consider quantization, nor did they incorporate a margin condition like Condition 2 into their formulation.
They considered a class of (inverse) likelihood ratio functions $\mathcal{F}$ that satisfies a complexity condition like Condition 1 and found that the difference ${\abbrefdiv{f^{*}}-\empirfdiv{\hat{f}_{\noobs}}}$ decays as $\mathcal{O}(\noobs^{-1/(2+\alpha)})$, where $D_n(\cdot)$ and $D(\cdot)$ are as defined in~\eqref{equ:empirKL} and~\eqref{equ:knownpq_KLest}, $f^*\in\mathcal{F}$ is the best in class likelihood ratio function, and $\hat{f}_n$ is an empirical estimator similar to~\eqref{equ:empirest2}.
Note that this rate is strictly less than the rate in Theorem~\ref{thm:esterrorconverge} even if $\marp$ is eliminated from the formulation (take $\marp\rightarrow\infty$).

\subsection{Approximation Error}\label{subsect:approxerror}
The approximation error analysis also requires that we now think of $\Phi$ as a class of piecewise constant functions (quantization rules) supported on RDPs.
As discussed in Section~\ref{subsect:RDP}, this is fully consistent with the definition given in~\eqref{equ:defcandclass2}.
With this in mind, we have the result:
\begin{theorem}[Approximation error]\label{thm:approxerrorconverge}
Let $\candclass(\asize,J,\lbound,\ubound)$, $\mapsopt$, and $\mapcopt$ be as defined in~\eqref{equ:defcandclass2},~\eqref{equ:knownpq_KLest}, and~\eqref{equ:mapcopt} respectively.
Suppose that Condition 1 is met for some constants $c$ and $C$.
Then the approximation error is bounded as
\begin{equation}\label{equ:vandegeerlem}
\abbrefdiv{\mapcopt}-\abbrefdiv{\mapsopt} \leq \textnormal{const}(\beta,\denlbound,\denubound,\lbound,\ubound,\asize)~2^{-\RDPdepth}.
\end{equation}
\end{theorem}
The proof of this result is given in Appendix~\ref{pf:approxerror} (see also~\cite{lexa_2011a}).
It follows a related, function estimation result in~\cite{ruicastro_phd_thesis2007} with one important exception: the KL divergence is not additive, thus unlike a mean squared error metric, the approximation error $\abbrefdiv{\mapcopt}-\abbrefdiv{\mapsopt}$ cannot be quantified cell by cell.  
Details are provided in the proof.

The combination of Theorem~\ref{thm:esterrorconverge} and Theorem~\ref{thm:approxerrorconverge} gives the decay rate of the total expected error in terms of the number of training samples $n$ and the depth $J$ of the uniform dyadic partition $\pi_{J}$.
To balance the errors and obtain a rate only in terms of $n$, one can express $J$ as a function of $n$.
Setting $J=\lceil \kappa \ln{n}/(2\kappa-1)\ln{2} \rceil$ yields the final result
\begin{equation}
\abbrefdiv{\mapcopt}-\expectv\abbrefdiv{\mapempest} \leq \textnormal{const}\cdot \noobs^{-\frac{\marp}{2\marp-1}},
\end{equation}
for sufficiently large $\noobs$.

\section{Application: Quantization under communication constraints}\label{sect:app}
When signals are measured and digitized at one location but processed at another, communication of the data is necessary.
Because of ever present power, computing, and rate constraints, the raw data cannot be transmitted in full fidelity; instead a summary of the data is sent.
When the ultimate goal is classification or detection, one strategy to maximize performance and minimize communication costs is to heavily quantize the data such that the KL divergence is maximized.
This is perhaps the simplest strategy and hence attractive when communications are severely constrained.  
Optimal likelihood-ratio partitions can be very different from typical nearest neighbor (Voronoi) partitions that are associated with quantizers designed to minimize mean squared error (see Figs.~\ref{fig:lr} and~\ref{fig:classification_maps}).
Nevertheless, past work in quantization for classification has forced a small-cell property in the design strategy resulting in partitions resembling nearest neighbor partitions~\cite{gupta_hero2003}.
Consequently, optimal partitions with disjoint regions, for example, cannot be well-approximated by these methods.
The EDM quantization method overcomes this shortcoming.
\begin{figure}
\centerline{\includegraphics[width=7cm]{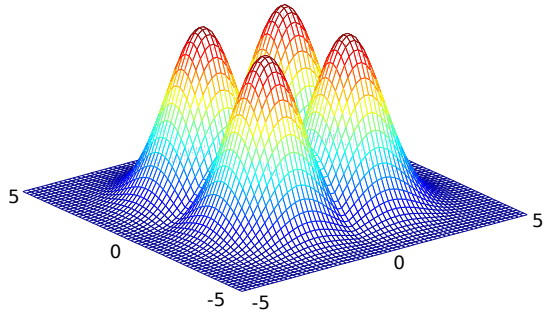}}
\caption{Plot of a likelihood-ratio function of a zero-mean bivariate Gaussian probability density function and a zero-mean bivariate Laplace probability density function.
The level sets of this function, which are concentric circles centered in the quadrants, form the boundaries of an optimal likelihood-ratio partition.
Such partitions are not well-approximated by optimal nearest neighbor partitions.}
\label{fig:lr}
\end{figure}

As an illustration, we consider $P$ to be a zero-mean bivariate Gaussian distribution and $Q$ to be a zero-mean bivariate Laplace distribution, both with identity correlation matrices; $P$ and $Q$ thus differ only in their basic shapes.
The plot of the likelihood ratio in Fig.~\ref{fig:lr} shows that the boundaries of the optimal likelihood-ratio partition are concentric circles in each quadrant.
Fig.~\ref{subfig:ex1a} depicts the best in class quantization rule along with its associated RDP in Fig.~\ref{subfig:ex1b}.
The result was generated with the Flynn and Gray algorithm but with $P_n(S_k)$ and $Q_n(S_k)$ in Algorithm~\ref{algor:flynn_gray} replaced by $P(S_k)$ and $Q(S_k)$.
(Data points lying outside of $[-5,5]^2$ were simply ignored.)
Convergence occurred in 8 iterations.
Fig.~\ref{subfig:ex1c} shows the empirical estimator generated from training sets each of size of two million samples.
In this case, the Flynn and Gray algorithm converged in 11 iterations.
\begin{figure*}
\centerline{\subfigure[Best-in-class rule]{\includegraphics[width=6cm]{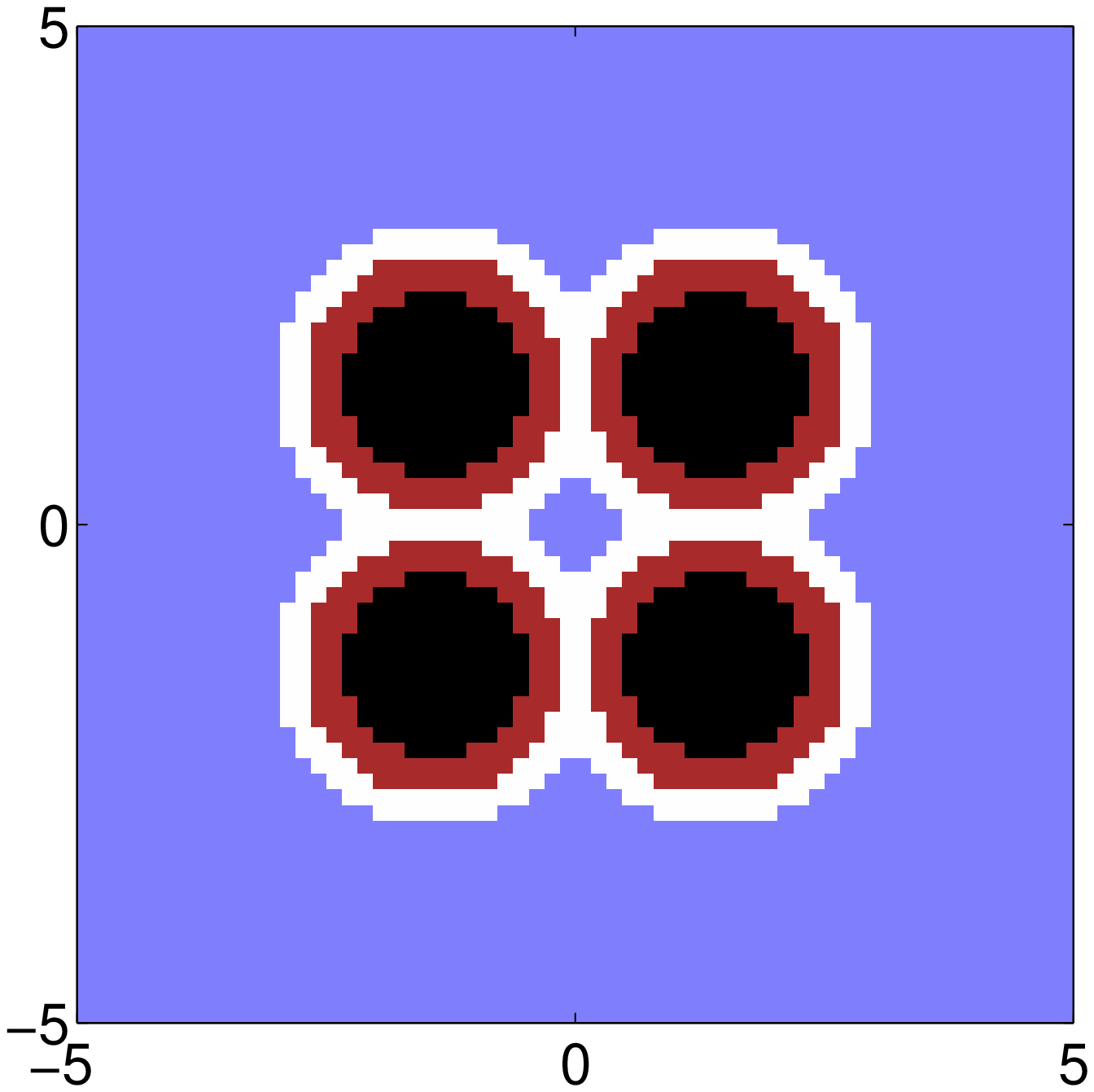}\label{subfig:ex1a}}
\hfil
\subfigure[Associated best-in-class RDP]{\includegraphics[width=6cm]{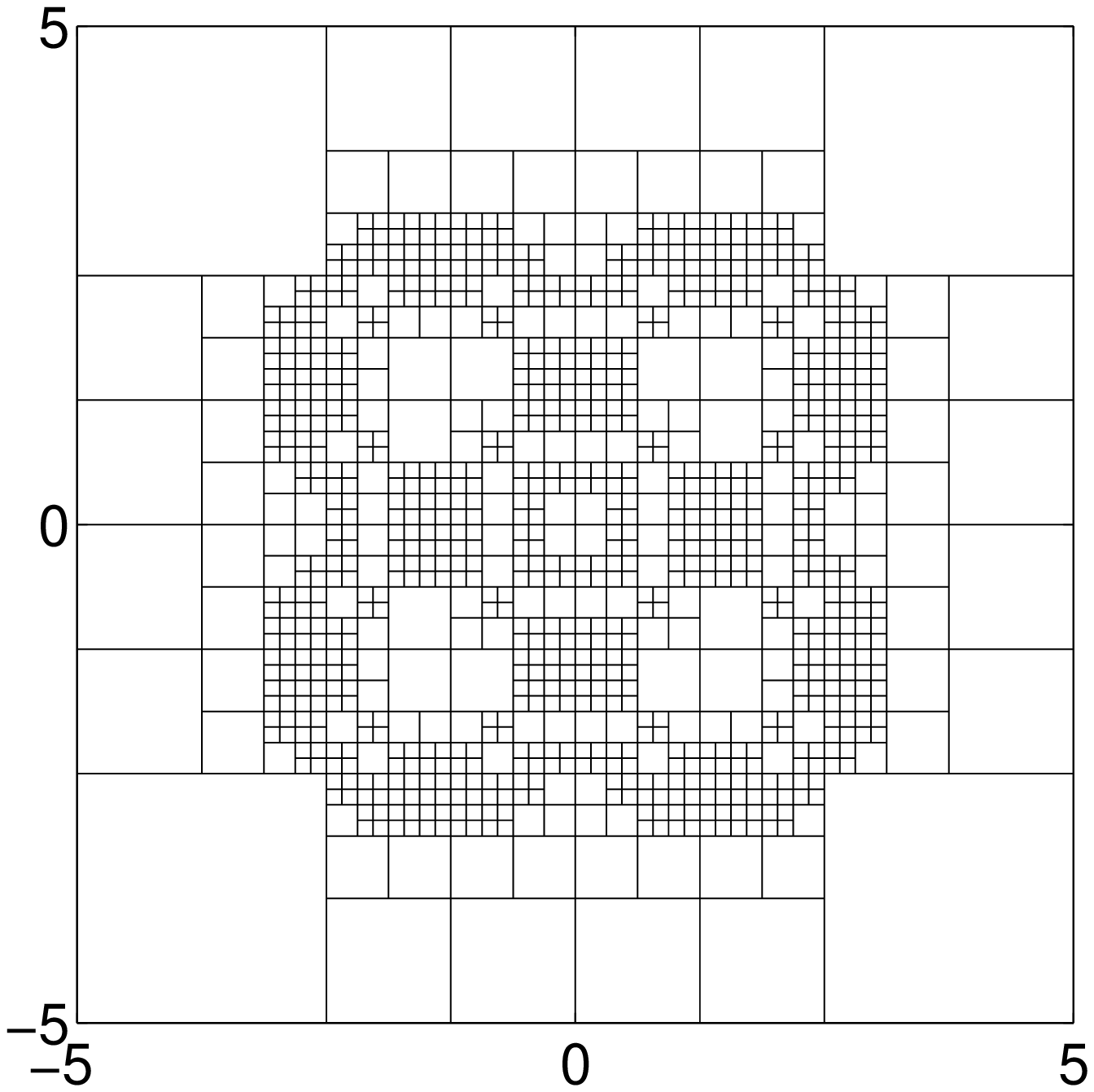}\label{subfig:ex1b}}}
\centerline{\subfigure[EDM estimator]{\includegraphics[width=6cm]{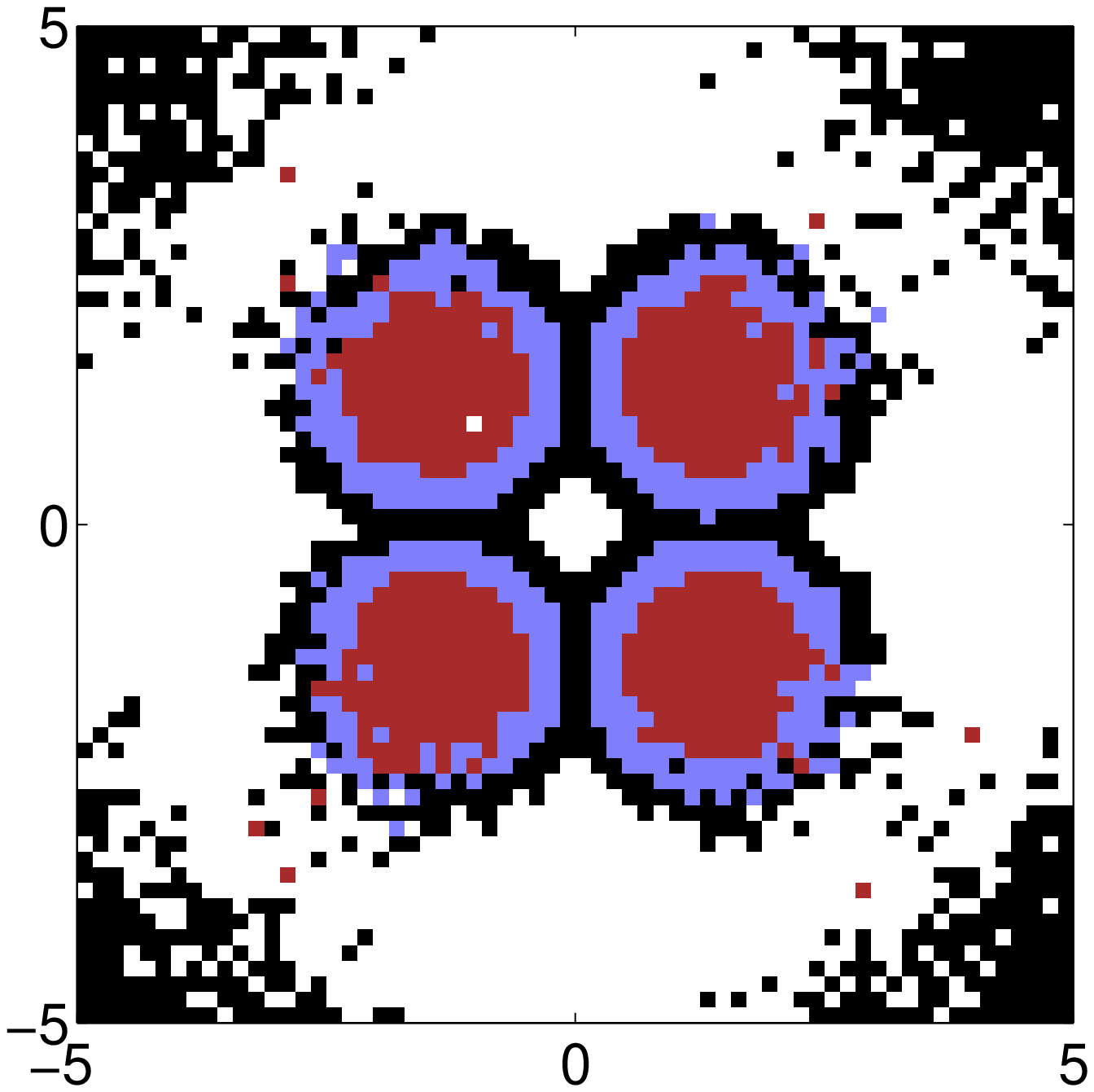}\label{subfig:ex1c}}
\hfil
\subfigure[Associated estimator RDP]{\includegraphics[width=6cm]{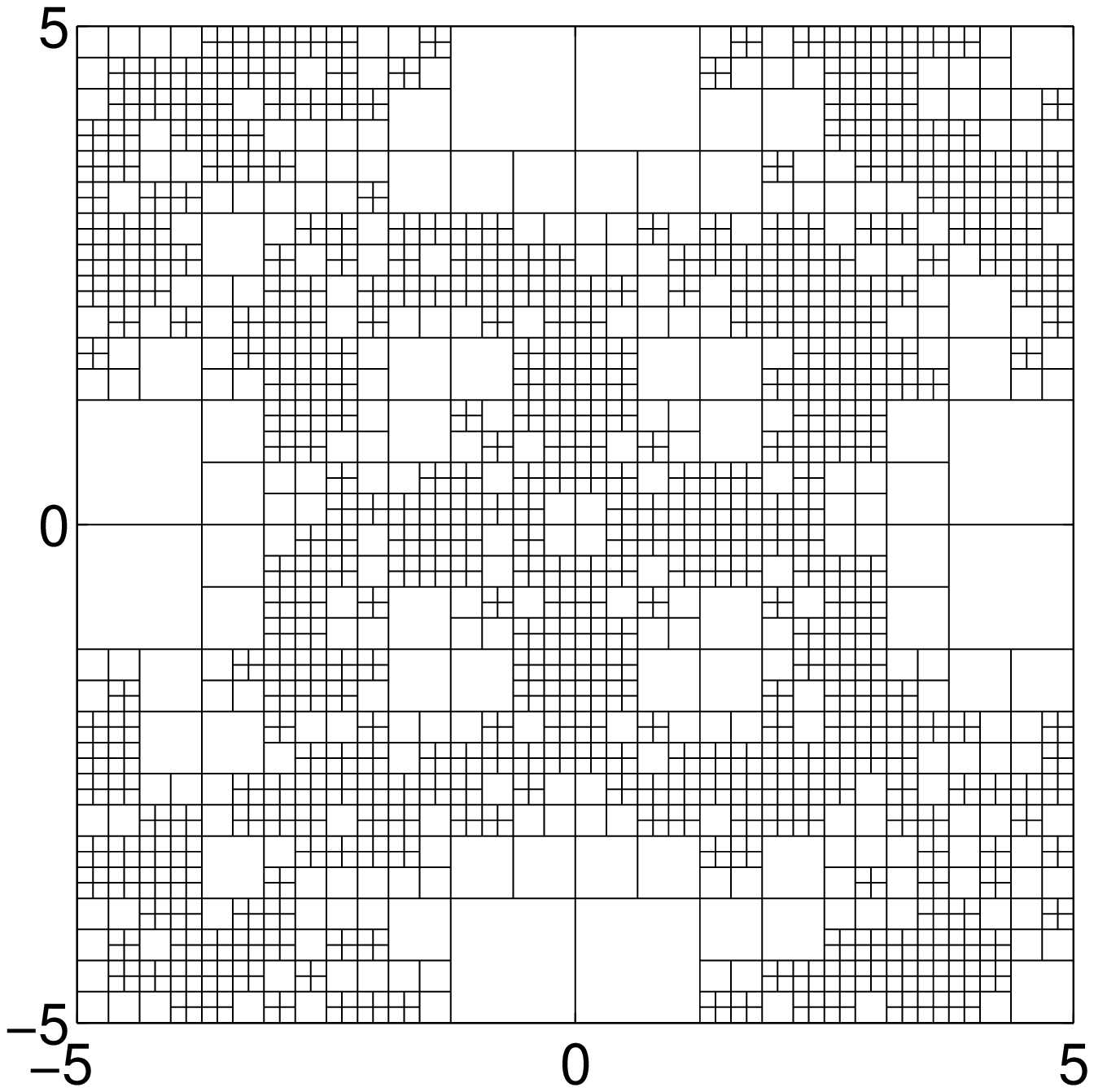}\label{subfig:ex1d}}}
\caption{Best-in-class and EDM quantization rules, and their associated recursive dyadic partitions when $P$ is bivariate Gaussian and $Q$ is bivariate Laplace, $L=4$, $J=6$. Note that the different cell labelings (colors) are inconsequential in terms of the divergence.}\label{fig:classification_maps}
\end{figure*}

In comparison to the best in class quantization rule, Fig.~\ref{fig:classification_maps} shows the effect of the trying to estimate $P$ and $Q$ on $\pi_{J}$ for low probability regions (corner regions).
In other words, the lack of data within these regions makes approximating $P$ and $Q$ on $\pi_{J}$ difficult, especially by empirical averages.
More sophisticated density estimation methods would improve this aspect of the estimator, such as kernal based methods.
The estimator might also be improved if one approximates $P$ and $Q$ on a (data-dependent) RDP instead of on $\pi_J$ (see e.g.,~\cite{devroye-gyorfi-lugosi1996}).

\section{Conclusion}\label{sect:conclusion}
In summary, EDM quantization provides a means of finding quantization rules, or more generally, low dimensional transformations that best preserve the divergence between two hypothesized distributions. 
EDM estimators can be computed using the Flynn and Gray algorithm, and they can exhibit fast error convergence rates as a function of the number of training samples.
The EDM formulation benefits from its connection to empirical process theory and possesses the flexibility to overcome the necessity of a small-cell constraint and allow efficient encoding.

\section*{Acknowledgments}
I thank John Thompson and Mike Davies for their financial support in presenting this work at conferences and for allowing me access to the computational resources of The University of Edinburgh. I also thank Don Johnson for many helpful discussions during the early stages of this work.

\appendix
\section{Appendices}

\subsection{Proof of Theorem~\ref{thm:esterrorconverge}}\label{app:esterror}
The proof proceeds by considering the behavior of $\rvert\genempir{\mapempest}-\genempir{\mapsopt}\lvert$ from~\eqref{equ:fundaineqal} as a function of the $L_2$ distance between $\mapempest$ and $\mapsopt$. 
This is done by considering a weighted empircal process for two different cases depending on the value of $\norm{\mapempest-\mapsopt}_{L_2}$.
The different cases yield different rates of convergence, thus they must be treated separately.
At the heart of the argument is Lemma 3, a concentration inequality result by van de Geer~\cite{vandegeer2000}, Lemma~5.13] concerning the supremum of weighted empirical processes (supremums are considered because we want uniform convergence).
The application of this result is not trivial, hence most of the proof is geared toward formulating the problem properly, most of this is done in Case 1 of the proof and Lemma 2.
Case 1 is a slight modification of a proof found in~\cite[pp. 206-207]{vandegeer2007}; Lemma 2 is original.
For more information regarding empirical process theory see~\cite{vandegeer2000}.
Lastly, as stated in Section~\ref{subsect:esterror}, the proof only requires the bracketing complexity of $\Phi$ satisfy $0<\alpha<2$.

Define the random variable
\begin{equation}\label{equ:RV}
\auxrv_{\noobs}=\frac{\rvert\genempir{\mapempest}-\genempir{\mapsopt}\lvert}{\denubound^{(\covcomp+2)/4}\norm{\mapempest-\mapsopt}_{L_2}^{1-\alpha/2}\vee \noobs^{-(2-\covcomp)/2(2+\covcomp)}},
\end{equation}
where $x\vee y=\max{(x,y)}$.
We consider the following two cases: $(1)~ \norm{\mapempest-\mapsopt}_{L_2}>\denubound^{-1/2} \noobs^{-1/(2+\covcomp)}$ and $(2)~ \norm{\mapempest-\mapsopt}_{L_2}\leq \denubound^{-1/2}\noobs^{-1/(2+\covcomp)}$.

\vspace{5pt}
\noindent\textbf{Case 1.}
Under this case~\eqref{equ:RV} simplifies to
\begin{equation}\label{equ:empprocesscond}
\auxrv_{\noobs}=\frac{\rvert\genempir{\mapempest}-\genempir{\mapsopt}\lvert}{\denubound^{(\covcomp+2)/4}\norm{\mapempest-\mapsopt}_{L_2}^{\compp}}
\end{equation}
where $\compp=1-\covcomp/2$.
For $\mapout=\mapsopt$,~\eqref{equ:empprocesscond} is defined to be zero.
Recalling the inequality~\eqref{equ:fundaineqal}, we have
\begin{align}
\begin{split}
\abbrefdiv{\mapcopt}-\abbrefdiv{\mapempest}&\leq -(\genempir{\mapsopt}-\genempir{\mapempest})/\sqrt{\noobs} \\
 & \qquad + \abbrefdiv{\mapcopt}-\abbrefdiv{\mapsopt}\end{split} \nonumber \\
\begin{split}\label{equ:vandegeer5}
&\leq \auxrv_{\noobs}\denubound^{(\covcomp+2)/4}\norm{\mapempest-\mapsopt}_{L_2}^{\compp}/\sqrt{\noobs} \\
& \qquad + \abbrefdiv{\mapcopt}-\abbrefdiv{\mapsopt}.
\end{split}
\end{align}
Condition 2 implies
\begin{equation}\begin{split}
\norm{\mapcopt-\mapempest}_{L_2}^{\compp}&\leq \marconst^{\compp/\marp} (\abbrefdiv{\mapcopt}-\abbrefdiv{\mapempest})^{\compp/\marp} \\
\norm{\mapcopt-\mapsopt}_{L_2}^{\compp}&\leq \marconst^{\compp/\marp} (\abbrefdiv{\mapcopt}-\abbrefdiv{\mapsopt})^{\compp/\marp}.
\end{split}\label{equ:vandegeer1}
\end{equation}
Hence, by the triangle inequality and~\eqref{equ:vandegeer1}, we have
\begin{align}
\norm{\mapempest-\mapsopt}_{L_2}^{\compp} &\leq \norm{\mapcopt-\mapempest}_{L_2}^{\compp} + \norm{\mapcopt-\mapsopt}_{L_2}^{\compp} \nonumber \\
\begin{split}
 &\leq \marconst^{\compp/\marp}(\abbrefdiv{\mapcopt}-\abbrefdiv{\mapempest})^{\compp/\marp} \\
 &\qquad + \marconst^{\compp/\marp}(\abbrefdiv{\mapcopt}-\abbrefdiv{\mapsopt})^{\compp/\marp}. \label{equ:vandegeer4}
\end{split}
\end{align}
Using~\eqref{equ:vandegeer4} in~\eqref{equ:vandegeer5}, shows $\abbrefdiv{\mapcopt}-\abbrefdiv{\mapempest}$ is less than or equal to
\begin{align*}
 &\Bigl[\denubound^{(\covcomp+2)/4}\marconst^{\compp/\marp}\noobs^{-1/2}\auxrv_{\noobs}(\abbrefdiv{\mapcopt}-\abbrefdiv{\mapempest})^{\compp/\marp}+ \denubound^{(\covcomp+2)/4} \\
 &\marconst^{\compp/\marp}\noobs^{-1/2}\auxrv_{\noobs}(\abbrefdiv{\mapcopt}-\abbrefdiv{\mapsopt})^{\compp/\marp}\Bigr]+ \abbrefdiv{\mapcopt}-\abbrefdiv{\mapsopt}.
\end{align*}
We now apply Lemma~\ref{lem:tech}~to each of the terms within the brackets to obtain
\begin{align*}
&\abbrefdiv{\mapcopt}-\abbrefdiv{\mapempest} \leq \epsilon \Bigl[(\abbrefdiv{\mapcopt}-\abbrefdiv{\mapempest})+(\abbrefdiv{\mapcopt}-\abbrefdiv{\mapsopt})\Bigr] \\
&\quad + 2\denubound^{\frac{-\marp\compp}{2(\marp+\compp)}}\Bigl(\frac{\marconst}{\epsilon}\Bigr)^{\frac{\compp}{\marp-\compp}} \auxrv_{\noobs}^{\frac{\marp}{\marp-\compp}}  \noobs^{-\frac{\marp}{2(\marp-\compp)}} + \abbrefdiv{\mapcopt}-\abbrefdiv{\mapsopt}.
\end{align*}
By rearranging the previous expression and dropping a factor of $\tfrac{1}{1+\epsilon}<1$, we have
\begin{align}\label{equ:vandegeer2}
&\abbrefdiv{\mapcopt}-\abbrefdiv{\mapempest} \leq  \Bigl(\frac{1+\epsilon}{1-\epsilon}\Bigr) \nonumber \\&
  \biggl[2 \denubound^{\frac{-\marp\compp}{2(\marp+\compp)}}(\marconst/\epsilon)^{\frac{\compp}{\marp-\compp}} \auxrv_{\noobs}^{\frac{\marp}{\marp-\compp}}  \noobs^{-\frac{\marp}{2(\marp-\compp)}}
  + \abbrefdiv{\mapcopt}-\abbrefdiv{\mapsopt}\biggr].
\end{align}
For any $r\geq \tfrac{\marp}{\marp-\compp}$, we have by Jensen's inequality~\cite{cover_thomas1991}, and Lemma~\ref{lem:empprocess} that
\begin{equation}\label{equ:vandegeer3}
\expectv \auxrv_{\noobs}^{\frac{\marp}{\marp-\compp}}
=\expectv\Bigl[(\auxrv_{\noobs}^{r})^{\frac{\marp}{r(\marp-\compp)}}\Bigr]
 \leq \bigl[\expectv\auxrv_{\noobs}^{r}\bigl]^{\frac{\marp}{r(\marp-\compp)}}
 \leq c_{2}^{\frac{\marp}{\marp-\compp}}.
\end{equation}
Taking the expectation of~\eqref{equ:vandegeer2} and applying~\eqref{equ:vandegeer3}, we conclude
\begin{equation*}\begin{split}
&\abbrefdiv{\mapcopt}-\expectv\abbrefdiv{\mapempest} \leq \Bigl(\frac{1+\epsilon}{1-\epsilon}\Bigr) \\
&\quad \biggl[2  \denubound^{\frac{-\marp\compp}{2(\marp+\compp)}}\bigl(\marconst/\epsilon\bigr)^{\frac{\compp}{\marp-\compp}} c_{2}^{\frac{\marp}{\marp-\compp}} \noobs^{-\frac{\marp}{2(\marp-\compp)}} + \abbrefdiv{\mapcopt}-\abbrefdiv{\mapsopt}\biggr],
\end{split}\end{equation*}
for $\norm{\mapempest-\mapsopt}_{L_2}> \denubound^{-1/2}\noobs^{-1/(2+\covcomp)}$.

\vspace{5pt}
\noindent\textbf{Case 2.}
For this case,
\begin{equation*}
\auxrv_{\noobs}=\frac{\rvert\genempir{\mapempest}-\genempir{\mapsopt}\lvert}{\noobs^{-(2-\covcomp)/2(2+\covcomp)}}.
\end{equation*}
From the fundamental inequality
\begin{align}
\begin{split}
\abbrefdiv{\mapcopt}-\abbrefdiv{\mapempest}&\leq -(\genempir{\mapsopt}-\genempir{\mapempest})/\sqrt{\noobs} \\
& \qquad + \abbrefdiv{\mapcopt}-\abbrefdiv{\mapsopt} \end{split}\nonumber \\
\begin{split}
&\leq \auxrv_{\noobs}\,\noobs^{-(2-\covcomp)/2(2+\covcomp)}\noobs^{-1/2} \\
& \qquad + \abbrefdiv{\mapcopt}-\abbrefdiv{\mapsopt} \end{split} \nonumber\\
&=\auxrv_{\noobs}\,\noobs^{-2/(2+\covcomp)} + \abbrefdiv{\mapcopt}-\abbrefdiv{\mapsopt}.  \label{equ:fundaineq2}
\end{align}
Taking the expectation of~\eqref{equ:fundaineq2} and applying Lemma~\ref{lem:empprocess} yields
\begin{equation}\label{equ:fundaineq3}
\abbrefdiv{\mapcopt}-\expectv\abbrefdiv{\mapempest}\leq c_{3}\,\noobs^{-2/(2+\covcomp)} + \abbrefdiv{\mapcopt}-\abbrefdiv{\mapsopt},
\end{equation}
for $\norm{\mapempest-\mapsopt}_{L_2}\leq \noobs^{-1/(2+\covcomp)}$.
The rate attained in~\eqref{equ:fundaineq3} is (strictly) faster than that attained in Case 1 (${\noobs^{-2/(2+\alpha)}<\noobs^{-\kappa/(2(\kappa-1)+\alpha)}}$ for $\marp>1-\alpha/2, \,0<\covcomp<2$).
Therefore, the decay of the total divergence loss is governed by the slower rate found in Case 1.
$\blacksquare$

\begin{lemma}[Tsybakov and van de Geer~\cite{tsybakov-vandegeer2005},van de Geer~\cite{vandegeer2007}] \label{lem:tech}
We have for all positive $v, t,\text{~and~}\epsilon$, and $\marp>\compp$,
\begin{equation*}
vt^{\compp/\marp}\leq \epsilon t + v^{\frac{\marp}{\marp-\compp}}\epsilon^{-\frac{\compp}{\marp-\compp}}.
\end{equation*}
\end{lemma}

\begin{lemma}\label{lem:empprocess}
Let $\mapempest$, $\mapsopt$, and $\mapcopt$ be as defined in~\eqref{equ:empirest2}, \eqref{equ:knownpq_KLest}, and~\eqref{equ:mapcopt} respectively.
Then under Conditions 1 and 2, we have
\begin{align}
&\expectv\left(\sup_{\mapout\in\candclass:\norm{\mapout-\mapsopt}_{L_2}> \denubound^{-1/2}\noobs^{-1/(2+\covcomp)}}
\frac{\rvert\genempir{\mapout}-\genempir{\mapsopt}\lvert}{\norm{\mapout-\mapsopt}_{L_2}^{1-\covcomp/2}}\right)^{r}\leq c_{2}^{r}, \label{equ:empprocess-a}\\
\intertext{and}
&\expectv\left(\sup_{\mapout\in\candclass:\norm{\mapout-\mapsopt}_{L_2}\leq \denubound^{-1/2}\noobs^{-1/(2+\covcomp)}}
\frac{\rvert\genempir{\mapout}-\genempir{\mapsopt}\lvert}{\noobs^{-(2-\covcomp)/2(2+\covcomp)}}\right)\leq c_{3}, \label{equ:empprocess-b}
\end{align}
for some positive constants $c_2, r,\text{and}~c_3$.
\end{lemma}

\begin{proof}
\noindent\textbf{Case 1:} Equation~\eqref{equ:empprocess-a}.
To compact notation, let $\natural$ denote the inequality ${\norm{\mapout-\mapsopt}_{L_2}> \denubound^{-1/2}\noobs^{-1/(2+\covcomp)}}$ and recall that
\begin{equation*}\begin{split}
\rvert\genempir{\mapout}-\genempir{\mapsopt}\lvert = 
\sqrt{\noobs}&\left\lvert\int(\log{\mapout}-\log{\mapsopt})~d(\probmeasureHzero_{\noobs}-\probmeasureHzero) \right. \\
& \quad \left.+ \int (\mapout-\mapsopt)~d(\probmeasureHone_{\noobs}-\probmeasureHone)\right\rvert.
\end{split}\end{equation*}
By Condition 1, we have
\begin{equation}
\begin{split}\label{equ:empprocess1}
\norm{\mapout-\mapsopt}_{L_2(\probmeasureHzero)}\leq \denubound^{1/2} \,\norm{\mapout-\mapsopt}_{L_2} \\
\norm{\mapout-\mapsopt}_{L_2(\probmeasureHone)}\leq \denubound^{1/2} \,\norm{\mapout-\mapsopt}_{L_2},
\end{split}
\end{equation}
for $\mapout\in\candclass$.
Consequently, we can write
\begin{align}
\sup_{\mapout\in\candclass:\, \natural}&\dfrac{\rvert\genempir{\mapout}-\genempir{\mapsopt}\lvert}{\norm{\mapout-\mapsopt}_{L_2}^{1-\covcomp/2}} \nonumber \\
&\leq \sup_{\mapout\in\candclass: \,\natural}\dfrac{\sqrt{\noobs}\left\lvert\displaystyle\int(\log{\mapout}-\log{\mapsopt})~d(\probmeasureHzero_{\noobs}-\probmeasureHzero)\right\rvert}{\norm{\mapout-\mapsopt}_{L_2}^{1-\covcomp/2}} \nonumber \\
&\qquad \quad + \sup_{\mapout\in\candclass:\, \natural} \dfrac{\sqrt{\noobs}\left\lvert\displaystyle\int(\mapout-\mapsopt)~d(\probmeasureHone_{\noobs}-\probmeasureHone)\right\rvert}{\norm{\mapout-\mapsopt}_{L_2}^{1-\covcomp/2}} \nonumber \\
\begin{split}\leq \sup_{\mapout\in\candclass:\, \natural}\dfrac{\sqrt{\noobs}\left\lvert\displaystyle\int(\log{\mapout}-\log{\mapsopt})~d(\probmeasureHzero_{\noobs}-\probmeasureHzero)\right\rvert}{\denubound^{(\covcomp-2)/4}\norm{\mapout-\mapsopt}_{L_2(\probmeasureHzero)}^{1-\covcomp/2}} \\
\qquad \quad + \sup_{\mapout\in\candclass:\, \natural} \dfrac{\sqrt{\noobs}\left\lvert\displaystyle\int(\mapout-\mapsopt)~d(\probmeasureHone_{\noobs}-\probmeasureHone)\right\rvert}{\denubound^{(\covcomp-2)/4}\norm{\mapout-\mapsopt}_{L_2(\probmeasureHone)}^{1-\covcomp/2}}
\end{split}\label{equ:empbreakup}
\end{align}
where the last inequality follows from $\eqref{equ:empprocess1}$.

We now want to apply a probability inequality due to van de Geer~\cite{vandegeer2000} (stated as Lemma~\ref{lem:vandegeer5-13} below) to the two terms in~\eqref{equ:empbreakup}.
The result requires $\candclass$ and $\widetilde{\candclass}=\{\log \phi: \phi\in\Phi\}$ to have a bracketing complexity satisfying $0<\alpha <2$.
$\Phi$ satisfies the requirement by construction which implies the same is true for $\widetilde{\Phi}$.
The result also requires that the differences $(\log \phi - \log \phi^*)$ and $(\phi-\phi^*)$ are upper bounded.
This follows from the definition of $\Phi$.
Furthermore, note that the proper form of the condition under the supremum follows from~\eqref{equ:empprocess1}.

Applying Lemma~\ref{lem:vandegeer5-13} to each term in~\eqref{equ:empbreakup}, we obtain
\begin{align*}
\Pr\Biggl(\sup_{\mapout\in\candclass: \, \natural}
&\dfrac{\sqrt{\noobs}\Bigl\lvert\displaystyle\int(\log{\mapout}-\log{\mapsopt})~d(\probmeasureHzero_{\noobs}-\probmeasureHzero)\Bigr\rvert}{\norm{\mapout-\mapsopt}_{L_2(\probmeasureHzero)}^{1-\covcomp/2}}\geq \denubound^{(\covcomp-2)/4}\, t \Biggr)\\
& \leq  \tilde{c}\,\exp{\biggl(-\frac{t}{c^2}\biggr)}
\end{align*}
and
\begin{align*}
\Pr\Biggl(\sup_{\mapout\in\candclass:\,\natural}
& \dfrac{\sqrt{\noobs}\Bigl\lvert\displaystyle\int(\mapout-\mapsopt)~d(\probmeasureHone_{\noobs}-\probmeasureHone)\Bigr\rvert}{\norm{\mapout-\mapsopt}_{L_2(\probmeasureHone)}^{1-\covcomp/2}}\geq \denubound^{(\covcomp-2)/4}\, t \Biggr) \\
& \leq \tilde{c}\,\exp{\biggl(-\frac{t}{c^2}\biggr)}
\end{align*}
for all $t\geq c$, some constant $\tilde{c}>0$, and $\noobs$ sufficiently large.
Consequently,
\begin{align}
\expectv \Biggr(\sup_{\mapout\in\candclass:\,\natural}\dfrac{\sqrt{\noobs}\left\lvert\displaystyle\int(\log{\mapout}-\log{\mapsopt})~d(\probmeasureHzero_{\noobs}-\probmeasureHzero)\right\rvert}{\norm{\mapout-\mapsopt}_{L_2(\probmeasureHzero)}^{1-\covcomp/2}}\Biggl)^r \leq c_{2,1} \nonumber \\
\intertext{and}
\expectv \Biggr(\sup_{\mapout\in\candclass:\, \natural}\dfrac{\sqrt{\noobs}\left\lvert\displaystyle\int(\mapout-\mapsopt)~d(\probmeasureHone_{\noobs}-\probmeasureHone)\right\rvert}{\norm{\mapout-\mapsopt}_{L_2(\probmeasureHone)}^{1-\covcomp/2}}\Biggl)^r \leq c_{2,2} \nonumber
\end{align}
for all $r>0$ and some finite positive constants $c_{2,1}$ and $c_{2,2}$.
Therefore
\begin{equation*}
\expectv\left(\sup_{\mapout\in\candclass:\, \natural}\frac{\rvert\genempir{\mapout}-\genempir{\mapsopt}\lvert}{\norm{\mapout-\mapsopt}_{L_2}^{1-\covcomp/2}}\right)^{r}\leq c_{2}^{r},
\end{equation*}
for some finite positive constant $c_2$.

\vspace*{10pt}
\noindent\textbf{Case 2:} Equation~\eqref{equ:empprocess-b}.
Let $\flat$ denote the inequality $\norm{\mapout-\mapsopt}_{L_2}\leq \denubound^{-1/2}\noobs^{-1/(2+\covcomp)}$.
From the definition of the empirical process $\rvert\genempir{\mapout}-\genempir{\mapsopt}\lvert$, we have
\begin{align}
\sup_{\mapout\in\candclass:\,\flat}&\frac{\rvert\genempir{\mapout}-\genempir{\mapsopt}\lvert}{\noobs^{-(2-\covcomp)/2(2+\covcomp)}}\nonumber \\
\begin{split}\leq \sup_{\mapout\in\candclass:\,\flat}\dfrac{\sqrt{\noobs}\left\lvert\displaystyle\int(\log{\mapout}-\log{\mapsopt})~d(\probmeasureHzero_{\noobs}-\probmeasureHzero)\right\rvert}{\noobs^{-(2-\covcomp)/2(2+\covcomp)}} \\
\qquad \quad + \sup_{\mapout\in\candclass:\,\flat} \dfrac{\sqrt{\noobs}\left\lvert\displaystyle\int(\mapout-\mapsopt)~d(\probmeasureHone_{\noobs}-\probmeasureHone)\right\rvert}{\noobs^{-(2-\covcomp)/2(2+\covcomp)}}.
\end{split}\label{equ:empbreakup2}
\end{align}
Apply Lemma~\ref{lem:vandegeer5-13} to each of the terms in~\eqref{equ:empbreakup2} to get
\begin{align*}
 \Pr\Biggl(\sup_{\mapout\in\candclass:\,\flat}
& \Bigl\lvert\displaystyle\int(\log{\mapout}-\log{\mapsopt})~d(\probmeasureHzero_{\noobs}-\probmeasureHzero)\Bigr\rvert\geq t\,\noobs^{-2/(2+\covcomp)} \Biggr) \\
& \leq c\,\exp{\Biggl(-\frac{t\,n^{\frac{\covcomp}{2+\covcomp}}}{c^2}\Biggr)}
\end{align*}
and
\begin{align*}
\Pr\Biggl(\sup_{\mapout\in\candclass:\,\flat}
& \Bigl\lvert\displaystyle\int(\mapout-\mapsopt)~d(\probmeasureHone_{\noobs}-\probmeasureHone)\Bigr\rvert\geq t\, \noobs^{-2/(2+\covcomp)} \Biggr) \\
& \leq c\,\exp{\Biggl(-\frac{t\,n^{\frac{\covcomp}{2+\covcomp}}}{c^2}\Biggr)},
\end{align*}
for all $t\geq c$ and for $\noobs$ sufficiently large.
Therefore, for $\noobs$ sufficiently large
\begin{equation*}
\expectv\left(\sup_{\mapout\in\candclass:\,\flat}\frac{\rvert\genempir{\mapout}-\genempir{\mapsopt}\lvert}{\noobs^{-(2-\covcomp)/2(2+\covcomp)}}\right)\leq c_{3},
\end{equation*}
for some positive constant $c_3$.
\end{proof}

\begin{lemma}[van de Geer~\cite{vandegeer2000}, Lemma~5.13]\label{lem:vandegeer5-13}
Let $X_1,\dotsc,X_n$ be an independent and identically distributed sequence of random variables on a probability space $(\mathcal{X},\mathcal{A},P)$.
Let $\mathcal{G}\subset L_2(P)$ be a collection of functions and define the empirical process indexed by $\mathcal{G}$ as
\begin{equation*}
\nu_n=\Bigl\{\nu_n(g)=\sqrt{n}\int g~d(P_n-P):~g\in\mathcal{G}\Bigr\}.
\end{equation*}
Let $\lvert g \rvert_{\infty}=\sup_{x\in\mathcal{X}}~\lvert g(x)\rvert$ denote the supremum norm
 and suppose $\sup_{g\in\mathcal{G}}~\lvert g-g_0 \rvert_{\infty} \leq K$, for some fixed element $g_0\in\mathcal{G}$ and some constant $K$.
Furthermore, suppose
\begin{equation*}
H_B(\delta,\mathcal{G},L_2(P))\leq A\delta^{-\rho}, \quad\text{for all~} \delta>0,
\end{equation*}
for some $0<\rho<2$ and some constant $A>0$.
Then for some constant $c$ depending on $\rho$ and $A$, we have for all $t\geq c$ and for $n$ sufficiently large,
\begin{equation*} \begin{split}
\Pr\Biggl(\sup_{g\in\mathcal{G},~\norm{g-g_0}\leq n^{-\frac{1}{2+\rho}}} &\left\lvert\int(g-g_0)d(P_n-P)\right\rvert\geq t\,n^{-\frac{2}{2+\rho}}\Biggr) \\ 
& \leq c\,\exp{\Biggl(-\frac{t\,n^{\frac{\rho}{2+\rho}}}{c^2}\Biggr)} 
\end{split} \end{equation*}
and
\begin{equation*} \begin{split}
\Pr\Biggr(\sup_{g\in\mathcal{G},~\norm{g-g_0}> n^{-\frac{1}{2+\rho}}} & \frac{\lvert\nu_n(g)-\nu_n(g_0)\lvert}{\norm{g-g_0}^{1-\frac{\rho}{2}}} \geq t \Biggr) \\
& \leq c\,\exp{\biggl(-\frac{t}{c^2}\biggr)},
\end{split}\end{equation*}
where the norms $\norm{g-g_0}$ are norms in $L_2(P)$.
\end{lemma}

\subsection{Proof of Theorem~\ref{thm:approxerrorconverge}}\label{pf:approxerror}
Recall the definition of $\PCclass$ from Section~\ref{subsect:approxerror}. 
We will need the following lemma.
\begin{lemma}[\cite{ruicastro_phd_thesis2007},~Lemma 5,~p.~121]\label{lem:rui}
There is a RDP such that the cells intersecting $B(\mapcopt)$ are at depth $\RDPdepth$ and all the other cells are at depths no greater than $\RDPdepth$.
Denote the smallest such RDP by $\partition^{*}_{\RDPdepth}$.
Then $\partition^{*}_{\RDPdepth}$ has at most $2^{2\ddimen}\beta 2^{(\ddimen-1)\RDPdepth}$ cells intersecting $B(\mapcopt)$.
\end{lemma}

Let $\mapout'$ denote the $\asize$-level piecewise constant function defined by,
\begin{equation}\label{equ:phiprimedef}
\mapout'(\inputsv)=\sum_{\partindex=0}^{\asize-1}\coeff_{\region'_{\partindex}} \indicator{\region'_{\partindex}}(\inputsv), \quad \coeff_{\region'_{\partindex}}=\frac{\probmeasureHzero(\region'_{\partindex})}{\probmeasureHone(\region'_{\partindex})}
\end{equation}
where each member region $\region'_{\partindex}$ of the partition $\{R_{i}'\}$ is composed of a union of cells $\cell\in \partition^{*}_{\RDPdepth}$.
Furthermore, let $\phi'$ satisfy the condition that the cells $\cell\in\partition^{*}_{\RDPdepth}$ contained in $\region_{i}'/B(\mapcopt)$ are also contained in $\optcell_{\partindex}$.
More concisely, we write $\cell \subseteq \region'_{\partindex}/B(\mapcopt) \Rightarrow \cell \subseteq \optcell_{\partindex}/B(\mapcopt)$.
In words, this last condition means that the partitions $\{\region'_{\partindex}\}$ and $\{\optcell_{\partindex}\}$ coincide except possibly on the boundary $B(\mapcopt)$.

First, observe that $\abbrefdiv{\mapcopt}-\abbrefdiv{\mapsopt}\leq \abbrefdiv{\mapcopt}-\abbrefdiv{\mapout'}$ since the divergence between the pmfs induced by $\mapout'$ is necessarily less than or equal to the that induced by the best in class quantization rule $\mapsopt$.
(This inequality also follows from the Data Processing Theorem~\cite[pp.~18-22]{kullback1959}.) 

Next, upper bound the difference $\abbrefdiv{\mapcopt}-\abbrefdiv{\mapout'}$ by the $L_{1}$-norm of $(\mapcopt-\mapsopt)$:
\begin{align}
\abbrefdiv{\mapcopt}-\abbrefdiv{\mapout'}&= \int_{[0,1]^{\ddimen}}\log{\frac{\mapcopt}{\mapout'}}~d\probmeasureHzero - \int_{[0,1]^{\ddimen}} (\mapcopt-\mapout')d\probmeasureHone  \nonumber \\
&\leq \int (\frac{\mapcopt}{\mapout'}-1)~d\probmeasureHzero - \int (\mapcopt-\mapout')d\probmeasureHone \nonumber \\
&= \int \frac{1}{\mapout'}(\mapcopt-\mapout')d\probmeasureHzero - \int (\mapcopt-\mapout')d\probmeasureHone \nonumber \\
&\leq   \frac{\denubound}{\lbound} \Big\lvert\int(\mapcopt-\mapout')d\inputsv\Big\rvert + \denlbound \Big\lvert\int(\mapcopt-\mapout')d\inputsv\Big\rvert \nonumber \\
&\leq \frac{\denubound+\denlbound\lbound}{\lbound}~\norm{\mapcopt-\mapout'}_{L_{1}},\label{equ:rubbish}
\end{align}
where the first inequality follows from the fact that $\log{x}\leq x-1$, for $x>0$, and the second inequality follows from the bounds on~$\mapout$, $\disHzero$, and $\disHone$.

Rewrite the $L_{1}$-norm as
\begin{align}
\norm{\mapcopt-\mapout'}_{L_{1}}&=\int_{[0,1]^{\ddimen}}\lvert \mapcopt(\inputsv)-\mapout'(\inputsv)\rvert~d\inputsv \nonumber \\
&=\sum_{\partindex=0}^{\asize-1} \int_{\region'_{\partindex}} \lvert\mapcopt(\inputsv)-\mapout'(\inputsv)\rvert~d\inputsv \nonumber \\
\begin{split} &= \sum_{\partindex=0}^{\asize-1} \Bigg[\sum_{\cell\subseteq\region'_{\partindex}/B(\mapcopt)} \int_{\cell}\lvert\mapcopt(\inputsv)-\mapout'(\inputsv)\rvert~d\inputsv  \\
& \quad + \sum_{\cell\subseteq\region'_{\partindex}(B(\mapcopt))} \int_{\cell}\lvert\mapcopt(\inputsv)-\mapout'(\inputsv)\rvert~d\inputsv \Bigg]\end{split}\label{equ:approx99}
\end{align}
Here, $\cell\subseteq\region'_{\partindex}/B(\mapcopt)$ means all cells $\cell$ that are a subset of $\region'_{\partindex}$ which do not intersect the boundary $B(\mapcopt)$.
Similarly, $\cell\subseteq\region'_{\partindex}(B(\mapcopt))$ means all cells $\cell$ that are subsets of $\region'_{\partindex}$ which do intersect $B(\mapcopt)$.

Consider the second summation within the brackets in~\eqref{equ:approx99}.
By the boundedness assumptions on $\mapcopt$ and $\mapout'$, the integrand can be upper bounded by $\ubound$.
Therefore,
\begin{align}
\sum_{\cell\subseteq\region'_{\partindex}(B(\mapcopt))} \int_{\cell}\lvert\mapcopt(\inputsv)&-\mapout'(\inputsv)\rvert~d\inputsv \nonumber \\
& \leq \ubound \sum_{\cell\subseteq\region'_{\partindex}(B(\mapcopt))} \text{Vol}(\cell) \nonumber \\
&\leq \ubound 2^{2\ddimen} \beta 2^{(\ddimen-1)\RDPdepth} 2^{-\ddimen\RDPdepth}  \nonumber \\
&=M 2^{2d}\beta \, 2^{-J},\label{equ:firstterm}
\end{align}
where the second inequality follows from Lemma~\ref{lem:rui} and the fact that the volume of one cell $S$ is $2^{-dJ}$.

Now, consider the first summation within the brackets in~\eqref{equ:approx99}.
For all $\cell\subseteq \region'_{\partindex}$ (and in particular for all $\cell\subseteq \region'_{\partindex}/ B(\mapcopt)$), $\mapout'$ equals $\probmeasureHone(\region'_{\partindex})/\probmeasureHzero(\region'_{\partindex})$ (recall~\eqref{equ:phiprimedef}).
Likewise, by the definition of $\phi'$, $\mapcopt$ is also constant for all $\cell\subseteq\region'_{\partindex}/B(\mapcopt)$.
Therefore, we have
\begin{align}
\sum_{\cell\subseteq\region'_{\partindex}/B(\mapcopt)} &\int_{\cell}\lvert\mapcopt(\inputsv)-\mapout'(\inputsv)\rvert~d\inputsv \nonumber \\
&=\sum_{\partindex=0}^{\asize-1} \Biggl\lvert \frac{\probmeasureHzero(\optcell_{\partindex})}{\probmeasureHone(\optcell_{\partindex})} - \frac{\probmeasureHzero(\region'_{\partindex})}{\probmeasureHone(\region'_{\partindex})} \Biggr\rvert  \sum_{\cell\subseteq\region'_{\partindex}/B(\mapcopt)}\text{Vol}(\cell) \nonumber \\
&\leq \sum_{\partindex=0}^{\asize-1} \Biggl\lvert \frac{\probmeasureHzero(\optcell_{\partindex})\probmeasureHone(\region'_{\partindex})-\probmeasureHzero(\region'_{\partindex})\probmeasureHone(\optcell_{\partindex})}{\probmeasureHone(\optcell_{\partindex})\probmeasureHone(\region'_{\partindex})} \Biggr\rvert \text{Vol}(\region'_{\partindex}) \nonumber \\ 
&\leq\displaystyle\frac{1}{\denlbound}\sum_{\partindex=0}^{\asize-1} \Biggl\lvert \frac{\probmeasureHzero(\optcell_{\partindex})\probmeasureHone(\region'_{\partindex})-\probmeasureHzero(\region'_{\partindex})\probmeasureHone(\optcell_{\partindex})}{\probmeasureHone(\optcell_{\partindex})} \Biggr\rvert. \label{equ:blah2}
\end{align}
Using the inequalities,
\begin{equation}\label{equ:blah}\begin{split}
&\probmeasureHone(\region'_{\partindex})\leq \probmeasureHone(\optcell_{\partindex}) +  \sum_{S\in R_{i}'(B(A^{*}_i))}Q(S)\\
&\probmeasureHzero(\region'_{\partindex})\geq \probmeasureHzero(\optcell_{\partindex}) - \sum_{S\in R_{i}'(B(A^{*}_i))}P(S), 
\end{split}\end{equation}
we upper bound each term in the summation in~\eqref{equ:blah2}
\begin{align*}
\frac{1}{\probmeasureHone(\optcell_{\partindex})} &\bigl\lvert \probmeasureHzero(\optcell_{\partindex})\probmeasureHone(\region'_{\partindex})-\probmeasureHzero(\region'_{\partindex})\probmeasureHone(\optcell_{\partindex}) \bigr\rvert \\
\begin{split}& \leq \frac{1}{\probmeasureHone(\optcell_{\partindex})} \Bigl\lvert P(A^{*}_i) \sum_{S\subseteq R_{i}'(B(A^{*}_i))}Q(S) \\
  &\quad + Q(A^{*}_i) \sum_{S\subseteq R_{i}'(B(A^{*}_i))}P(S) \Bigr\rvert\end{split} \\
&\leq \frac{1}{\probmeasureHone(\optcell_{\partindex})} \bigl\lvert \probmeasureHzero(\optcell_{\partindex}) \denubound \beta' 2^{-\RDPdepth} + \probmeasureHone(\optcell_{\partindex}) \denubound \beta' 2^{-\RDPdepth} \bigr\rvert \\
&= \denubound \beta' 2^{-\RDPdepth} \Bigl\lvert\frac{\probmeasureHzero(\optcell_{\partindex})+\probmeasureHone(\optcell_{\partindex})}{\probmeasureHone(\optcell_{\partindex}) } \Bigr\rvert, 
\end{align*}
where the second inequality follows from Lemma~\ref{lem:rui} with $\beta'=2^{2d}\beta$.

Summarizing, we have
\begin{align}
\sum_{\partindex=0}^{\asize-1} &\sum_{\cell\subseteq\region'_{\partindex}/B(\mapcopt)} \int_{\cell}\lvert\mapcopt(\inputsv)-\mapout'(\inputsv)\rvert~d\inputsv \nonumber\\
&\leq \frac{\denubound}{\denlbound} \beta' 2^{-\RDPdepth} \sum_{\partindex=0}^{\asize-1}\Bigl\lvert\frac{\probmeasureHzero(\optcell_{\partindex})+\probmeasureHone(\optcell_{\partindex})}{\probmeasureHone(\optcell_{\partindex}) } \Bigr\rvert \nonumber\\
& \leq \frac{C}{c}\Bigl(\frac{C}{c}+1\Bigr) \beta' \asize\, 2^{-\RDPdepth}, \label{equ:doublestar}
\end{align}
where the last step follows from the assumed bounds on $p$ and $q$.

Finally, by combining~\eqref{equ:rubbish}, \eqref{equ:approx99}, \eqref{equ:firstterm}, and \eqref{equ:doublestar}, we conclude
\begin{equation}
\norm{\mapcopt-\mapout'}_{L_{1}}\leq \beta' L \bigl[M+(C/c)(C/c+1)\bigr] \, 2^{-J}
\end{equation}
and
\begin{equation}\begin{split}
&\abbrefdiv{\mapcopt}-\abbrefdiv{\mapout'} \leq \\
& \qquad \Bigl(\frac{\denubound+\denlbound\lbound}{\lbound}\Bigr) \beta' L \bigl[M+(C/c)(C/c+1)\bigr] \, 2^{-J}
\end{split}\end{equation}
$\blacksquare$

\single
\bibliographystyle{IEEEbib}
\bibliography{lexa}

\end{document}